%% file: arxiv-main.tex
\documentclass[reprint,amsmath,amssymb,aps,superscriptaddress,nofootinbib,pra,floatfix]{revtex4-2}

\usepackage{amsthm}
\usepackage[english]{babel}
\usepackage{mathtools}
\usepackage{physics}
\usepackage{xcolor}
\usepackage{graphicx}
\usepackage{adjustbox}
\usepackage{placeins}
\usepackage[T1]{fontenc}
\usepackage[utf8]{inputenc}
\usepackage{csquotes}
\usepackage{dcolumn}
\usepackage{subfigure}
\usepackage{comment}
\usepackage{booktabs}
\usepackage{tabularx}
\usepackage{enumerate}
\usepackage{enumitem}
\usepackage[colorinlistoftodos, color=green!40, prependcaption]{todonotes}
\usepackage[pdftex, pdftitle={Article}, pdfauthor={Author}]{hyperref} 
\hypersetup{
  colorlinks=True,
  citecolor=blue,
  linkcolor=blue,
  urlcolor=blue}
\usepackage{makecell}
\usepackage{multirow}
\usepackage{times}
\usepackage[normalem]{ulem}

\newtheorem{definition}{Definition}
\newtheorem{theorem}{Theorem}

\newtheorem{lemma}{Lemma}

\begin{document}
\title{Fully passive quantum random number generation with untrusted light}

\author{KaiWei Qiu}
    \affiliation{School of Physical and Mathematical Sciences, Nanyang Technological University, Singapore}

\author{Yu Cai}
    \affiliation{School of Physical and Mathematical Sciences, Nanyang Technological University, Singapore}

\author{Nelly H.Y. Ng}
    \email{nelly.ng@ntu.edu.sg}
    \affiliation{School of Physical and Mathematical Sciences, Nanyang Technological University, Singapore}
    \affiliation{Centre for Quantum Technologies, National University of Singapore}

\author{Jing Yan Haw}
    \email{jingyan.haw@nus.edu.sg}
    \affiliation{Centre for Quantum Technologies, National University of Singapore}

\begin{abstract}
Quantum random number generators (QRNGs) harness the inherent unpredictability of quantum mechanics to produce true randomness. Yet, in many optical implementations, the light source remains a potential vulnerability — susceptible to deviations from ideal behavior and even adversarial eavesdropping. Source-device-independent (SDI) protocols address this with a pragmatic strategy, by removing trust assumptions on the source, and instead rely on realistic modelling and characterization of the measurement device. In this work, we enhance an existing SDI-QRNG protocol by eliminating the need for a perfectly balanced beam splitter within the trusted measurement device, which is an idealized assumption made for the simplification of security analysis. We demonstrate that certified randomness can still be reliably extracted across a wide range of beam-splitting ratios, significantly improving the protocol’s practicality and robustness. Using only off-the-shelf components, our implementation achieves real-time randomness generation rates of $0.347$ Gbps. We also experimentally validate the protocol's resilience against adversarial attacks and highlight its self-testing capabilities. These advances mark a significant step toward practical, lightweight, high-performance, fully-passive, and composably secure QRNGs suitable for real-world deployment.
\end{abstract}

\maketitle

\section{Introduction} \label{sec:Introduction}

Quantum random number generators (QRNGs) leverage the intrinsic probabilistic nature of quantum theory to generate genuine randomness~\cite{RevModPhys.94.025008}. A handful of these devices operate based on principles of quantum optics, leveraging light as the primary source of randomness and using  photodetection devices to extract quantum entropy from the optical signals~\cite{Herrero_Collantes_2017, ma2016quantum}. In theory, the comprehensive knowledge of a QRNG's internal design, encompassing details of the light source used and measurements, would ensure that the extracted randomness is unpredictable to potential adversaries. Yet, achieving a full, real-time characterization is technically challenging and often entails significant costs. Therefore, online certification of high-performance QRNGs is an important and critical issue in the development of such devices. 

Depending on the assumptions based on which security is derived, QRNG certification methods are typically categorized as device independent (DI), semi-DI, or device dependent (DD). DI-QRNGs provide security with minimal assumptions, typically certified through Bell inequality violations~\cite{bell1964einstein}, but they require complex experimental setups and generally yield low generation rates~\cite{acin2016certified, liu2018device, pironio2010random, Bierhorst2018}. On the other hand, DD-QRNGs assume full knowledge and trust in the entire experimental setup~\cite{gabriel2010generator, haw2015maximization, zheng20196, christianQRNG, bruynsteen2023100, Nie68Gbps}, demanding high stability and precise control, which can be impractical in real-world deployments. The intermediate semi-DI regime offers a more practical balance between security and implementation complexity. In this approach, partial assumptions are made — such as trusting or characterizing either the light source~\cite{drahi2020certified, cao2016source, marangon2017source, smith2019simple, cheng2022mutually, michel2019real, avesani2018source, zhang2023realization}, the measurement device~\cite{cao2015loss, chaturvedi2015measurement, nie2016experimental, wang2023provably}, the dimension of the Hilbert space~\cite{lunghi2015self, mironowicz2021quantum}, or energy constraints on the emitted photons~\cite{rusca2020fast, avesani2021semi}. This enables simpler QRNG designs with high random number generation rates and strong security, provided that the device passes the appropriate certification tests.

From a practical standpoint, source-device-independence is particularly crucial because the entropy source forms the foundation of randomness in a QRNG. Any compromise in the integrity of the source directly undermines the security and reliability of the generated random numbers. To address this, source-device-independent (SDI) QRNGs relax the trust assumptions on the light source, assuming that it could be entirely controlled by a malicious adversary, Eve, while relying only on trusted and well-characterized measurement devices~\cite{drahi2020certified, cao2016source, marangon2017source, smith2019simple, cheng2022mutually, michel2019real, avesani2018source, zhang2023realization}. In this setting, the output randomness can still be certified as truly random and close to uniform after appropriate post-processing. 

Remarkably, Ref.~\cite{drahi2020certified} proposed and demonstrated a composable, high-speed (Gbps) continuous variable SDI-QRNG protocol based on a totally untrusted photonic source. While its experimental setup shared similarities with continuous-variable (CV) QRNGs employing balanced homodyne detector(s)~\cite{gabriel2010generator,haw2015maximization,zheng20196,christianQRNG,bruynsteen2023100}, the key distinction lies in the fact that there is no requirement to trust or characterize a local oscillator. This is because the protocol extracts randomness from the difference measurement between the photodetectors, rather than from a quadrature measurement typical in homodyne detection—where a strong local oscillator is treated as part of the trusted measurement device. 

However, in any practical implementation, measurement imperfections and side channels inevitably arise. These imperfections, in principle, can be exploited by an adversary using quantum hacking techniques—similar to those observed in CV quantum key distribution systems~\cite{huang2020practical, chi2011balanced, PhysRevA.84.062308, PhysRevA.87.052309, PhysRevA.98.012312}. Consequently, even within the SDI paradigm, it is desirable to minimize the assumptions and experimental requirements imposed on the measurement apparatus. In view of this, we note that one of the assumptions for the protocol proposed in Ref.~\cite{drahi2020certified} is that the difference measurement device uses a perfectly balanced (i.e. 50:50) optical beam splitter. In practice, perfect balancing is unattainable due to manufacturing imperfections and finite optical path length differences. As such, without taking this realistic imperfection into account, the aforementioned SDI-QRNG protocol will overestimate the amount of randomness generated, leading to potential security loopholes in the QRNG. 

In this paper, we first extend the security proof of the SDI protocol in Ref.~\cite{drahi2020certified} to accommodate the presence of an unbalanced optical beam splitter. This enhancement strengthens the security of the protocol under practical imperfections and mitigates a critical assumption in existing implementations. By doing so, our extended analysis alleviates the need for active balancing components, such as variable optical attenuators (VOAs) or variable optical delays (VODs), thereby simplifying the experimental setup and reducing overall system complexity. Building upon this theoretical foundation, we demonstrate the feasibility of our protocol through the development of a lightweight and cost-effective SDI-QRNG prototype, constructed entirely from off-the-shelf components. Notably, our system operates in real-time and performs randomness certification and extraction without requiring a perfectly balanced optical beam splitter, thus affirming its practicality and robustness. Finally, we evaluate the security of our protocol under conditions of intensity fluctuation, simulating an adversarial scenario where untrusted light may be injected into the QRNG system. This experimental validation further underscores the resilience of our SDI-QRNG protocol against real-world implementation vulnerabilities. 

\section{SDI-QRNG Framework} \label{sec:Protocol}

An SDI-QRNG protocol~\cite{drahi2020certified} (see Fig.~\ref{fig: SDI-QRNG Schematic}) consists of three components: (1) An untrusted light source, which is assumed to be fully controlled by Eve, (2) Randomness generation using trusted and reliably characterized measurement devices, including optical beam splitters, vacuum inputs and photodetectors, and (3) Randomness extraction protocol to extract the final random numbers that are close to being uniform and uncorrelated from Eve. We note that both the randomness generation and extraction stages are essential for the QRNG device. Without the latter stage, the device is called a quantum randomness generator (QRG) that acts as a source of raw quantum entropy, where the output could still be non-uniform and correlated to Eve. 

\begin{figure*}[t!]
    \centering
    \includegraphics[width=0.8\linewidth]{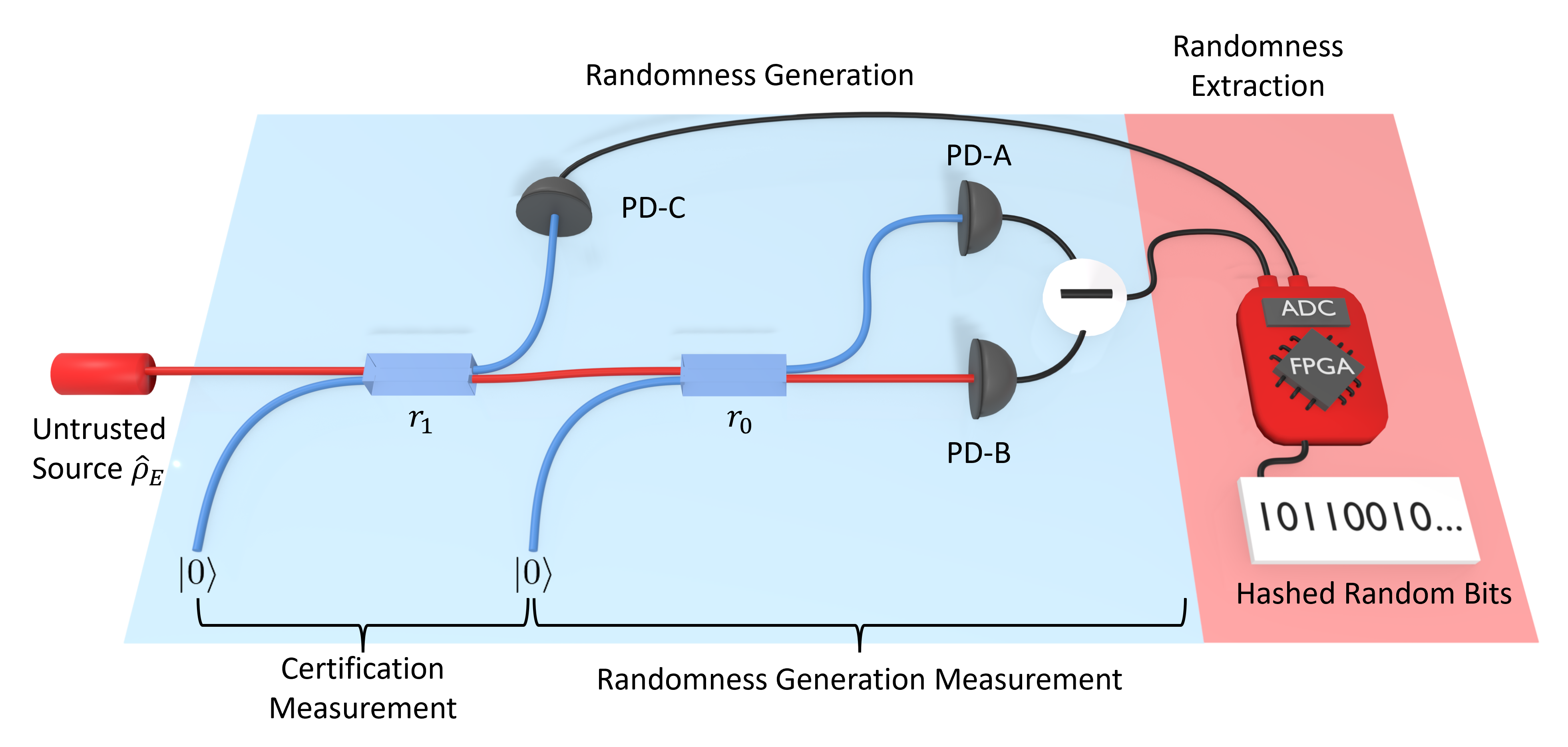}
    \caption{The schematic of the SDI-QRNG setup. An untrusted light $\hat{\rho}_E$ (assumed to be fully controlled by an eavesdropper, Eve) enters the QRG, where a fiber beam splitter of reflectivity $r_1$ reflects some of the light for the certification measurement at PD-C. The remaining light enters the randomness generation measurement stage with a fiber beam splitter of reflectivity $r_0$. Upon passing the test, randomness is generated via a difference measurement between PD-A and PD-B. Finally, the randomness output is sent for randomness extraction to produce hashed random bits which are close to being uniformly random with respect to Eve. PD: photodetector.}
    \label{fig: SDI-QRNG Schematic}
\end{figure*}

The total randomness of the classical outcome $X$ produced by randomness generation is quantified by the min-entropy of $X$ conditioned on Eve's knowledge $E$, denoted by $H_{\min}(X|E)$. This includes knowledge about the light source and measurement devices, which can in general be stored, e.g. in a quantum memory. This results in a classical-quantum state $\hat{\rho}_{XE} = \sum_x p_x \ket{x}\bra{x} \otimes \hat{\rho}_E^x$ for the joint system $XE$, where $p_x$ is the probability of $X = x$ occurring and $\hat{\rho}_E^x$ is the density operator of the state of $E$ conditional on $X=x$ \cite{konig2009operational}. It is known that randomness is adequately quantified by the conditional min-entropy,
\begin{eqnarray}
    H_{\min}(X|E) := -\log_2\left(\sup_{\{\hat{E}_x\}}\sum_x p_x \text{tr}\left(\hat{E}_x \hat{\rho}_E^x\right)\right),
\end{eqnarray}
where the supremum is over all possible POVM measurements of $\{\hat{E}_x\}$ on Eve. The term within the logarithm corresponds to Eve's best guessing probability of outcome $X$. 

\subsection{Randomness Generation}\label{sec: Randomness Generation}

We first establish a formal security definition for a certifiable randomness generation protocol. Similar to the security definition for the quantum key distribution protocol \cite{RevModPhys.94.025008}, the randomness generation protocol comprises of a \textit{Security} aspect that ensures its security with a certification test $\mathcal{P}$ for generating certified randomness. The protocol is then aborted if the test is failed. The protocol also need a \textit{Completeness} aspect to ensure that the test $\mathcal{P}$ is consistently passed with high probability under an honest implementation. The overall security of the protocol must also be composable. Formally, the certifiable QRG protocol is presented as the following~\cite{drahi2020certified}.

\begin{definition}\label{def: Definition 1}
An ($m$, $\kappa$, $\epsilon_{\textnormal{fail},m}$, $\epsilon_C$)-certified randomness generation protocol produces output $X$ made of $m$ measurement results such that
\begin{enumerate}
    [leftmargin=*]\itemsep0em
    \item \textbf{Security:} Either the certification test $\mathcal{P}$ fails, or 
    $$H_{\min}(X|E)\geq\kappa$$
    except with probability $\epsilon_{\text{fail},m}$.
    \item \textbf{Completeness:} There exists an honest implementation that passes the test $\mathcal{P}$ with probability $1-\epsilon_C$.
\end{enumerate}
\end{definition}

The SDI protocol consists of two processes: certification measurement and randomness generation measurement, as depicted in Fig.~\ref{fig: SDI-QRNG Schematic}. For $m=1$ round of measurement, the SDI protocol begins with an untrusted light source $\hat{\rho}_E$ entering the QRG. It will be mixed with a trusted vacuum state $\ket{0}$ at the optical beam splitter with reflectivity $r_1$. The reflected light will undergo a certification measurement, where a certification test $\mathcal{P}$ ensures that the number of photons entering photodetector C falls within the photon range $n_C \in [n_C^-,n_C^+]$ with a passing probability of $1-\epsilon_C$. If the test fails, the protocol will abort, and a new measurement round begins. This certification test $\mathcal{P}$ ensures that the remaining untrusted light entering the randomness generation measurement is certified and will fall in a range $n_R\in[n_R^-,n_R^+]$, except with a failure probability $\epsilon_{\text{fail}}$. In the event of successful certification, the transmitted light shall subsequently be mixed with trusted vacuum at a second beam splitter, with reflectivity $r_0$. The reflected and transmitted light are then measured by photodetectors A and B, respectively. The random bit string $X$, which corresponds to the difference in the number of photons between the two photodetectors, will have a particular conditional min-entropy, denoted by $H_{\min}^{\text{SDI}}(X|E)$, for randomness extraction. We summarize the flow of the protocol in Table~\ref{tab: SDI protocol}.

\begin{table}
\caption{A flow chart summarizing the SDI protocol.}
\label{tab: SDI protocol}
\begin{ruledtabular}
\begin{tabular}{p{0.98\linewidth}}
\textbf{SDI Protocol Flow Chart}\\ \midrule
\begin{enumerate}[leftmargin=*]
    \item \textit{Source.} For $m=1$ round of measurement, an untrusted light source $\hat{\rho}_E$ enters the QRG.
    \item \textit{Certification.} The light undergoes the certification measurement, where a certification test $\mathcal{P}$ ensures $n_C \in[n_C^-,n_C^+]$ with a passing probability of $1-\epsilon_C$. Else, the protocol is aborted.
    \item \textit{Randomness Generation.} Upon passing test $\mathcal{P}$, the remaining photons entering the randomness generation measurement, $n_R$, will be certified except with a failure probability of $\epsilon_{\text{fail}}$ if $n_R\notin[n_R^-,n_R^+]$.
    \item \textit{Certified Min-Entropy.} The randomness is generated via a difference measurement, where they will be have a particular $H^{\text{SDI}}_{\text{min}}(X|E)$ for randomness extraction. 
\end{enumerate}
\end{tabular}
\end{ruledtabular}
\end{table}

\begin{table*}
\caption{Extended SDI Protocol}
\begin{ruledtabular}\label{tab: Extended SDI protocol}
\begin{tabular}{p{0.99\linewidth}}
\textbf{Extended SDI Protocol}\\ \midrule
 An optical setup consisting of 
   \begin{enumerate}[leftmargin=*]
       \item two trusted vacuum modes
       \item two fiber beam splitters of arbitrary reflectivity $r_0$ and $r_1$
       \item two noisy photodetectors (A and B) used to make a difference measurement 
       \item a third noisy photodetector used to make a certification measurement which passes the certification test $\mathcal{P}$ if the voltage bin of the Analog-to-Digital Converter (ADC) at photodetector C, $i_C$, falls in the chosen bin range of $[i_C^-,i_C^+]$.
   \end{enumerate}
   can be used as a certified ($m,\kappa$, $\epsilon_{\text{fail},m}$, $\epsilon_C$)-randomness generation protocol, satisfying 
   
\begin{enumerate}
[leftmargin=*]\itemsep0em
   \item \textbf{Security:} the randomness obtained is given by
   \begin{eqnarray}\label{eq: Hmin definition}
        H_{\min,r_0}^{\textnormal{SDI}}(X|E)\geq \kappa \geq -m\log_2\left[\sum_{x\in\mathcal{X}^{\textnormal{SDI}}_{r_0}}r_{0}^{\left\lfloor\frac{n_R^-+x}{2}\right\rfloor}(1-r_0)^{\left\lceil\frac{n_R^--x}{2}\right\rceil} \left(\begin{array}{c}n_R^{-} \\ \left\lfloor\frac{n_R^{-}+x}{2}\right\rfloor\end{array}\right)\right] 
   \end{eqnarray}
   where
   \begin{equation}
    \mathcal{X}^{\textnormal{SDI}}_{r_0} \in \mathbb{Z} \cap \left[\mu_x - \left\lceil\frac{\delta V}{2 \alpha_D} \right\rceil , \mu_x + \left\lfloor \frac{\delta V}{2 \alpha_D} \right\rfloor \right]
   \end{equation}
   with $\mu_x = 2 \lceil(n_R^-+1)r_0-1\rceil -n_R^-$, $\delta V_D = (V_{\max}^D-V_{\min}^D)/2^{\Delta_{\textnormal{ADC}}}$ and $\Delta_{\textnormal{ADC}}$ is the effective number of bits (ENOB) of the ADC.
   For $m$-rounds of measurement, the security parameter of the protocol, $\epsilon_{\text{fail},m}$, is $\epsilon_{\textnormal{fail,m}} \leq m\cdot \epsilon_{\textnormal{fail}}$,
   where the failure probability of a single round is
    \begin{eqnarray}\label{eq: efail}
       \epsilon_{\textnormal{fail}} &=& \max_{\hat{\rho}_E} \textnormal{Pr} \left(i^-_C \leq i_C\leq i^+_C \And n_R \notin [ n_R^-,n_R^+]\right)  = \max\{\epsilon_-,\epsilon_+\}+ \epsilon_{\gamma_C},\\ \label{eq: epsilon- and epsilon+}
    \epsilon_- &\leq& \sum_{n_C=n_C^-}^{n_{E}^{-}}\frac{r_1^{n_C}(1-r_1)^{n_{E}^{-}-n_C}n_{E}^{-}!}{n_C!(n_{E}^{-}-n_C)!} ,\qquad
    \epsilon_+ \leq \sum_{n_R=n_R^+}^{n_{E}^{+}}\frac{(1-r_1)^{n_R}(r_1)^{n_{E}^{+}-n_R}n_{E}^{+}!}{n_R!(n_{E}^{+}-n_R)!} , \\ 
    \epsilon_{\gamma_C} &=& 1-\textnormal{erf}\left(\frac{\Tilde{\gamma}_C}{\sqrt{2}\sigma_{\gamma_C}}\right),
   \end{eqnarray}
    where $n_{E}^{\pm}= n_C^{\pm}+n_R^{\pm}\pm1$, $n_R^+$ is set to the saturating photon number of the difference measurement, and $\gamma_C$ is the electronic noise variable of the certification photodetector such that $\abs{\gamma_C}<\Tilde{\gamma}_C$ except with probability $\epsilon_{\gamma_C}$.
    
    \item \textbf{Completeness:} There exists an honest implementation with coherent state $\ket{\alpha}$ as input for this SDI-QRNG, such that the certification test $\mathcal{P}$ has a passing probability of
    \begin{equation} \label{eq: completeness of coherent state}
        1-\epsilon_C = \textnormal{tr}\left\{\sum_{i_C=i_C^-}^{i_C^+} \ket{\alpha}\bra{\alpha}\hat{V}_C^{\sigma_{\gamma_C},\Delta_{\textnormal{ADC}}}(i_C)\right\}.
    \end{equation}
    \end{enumerate}
\end{tabular}
\end{ruledtabular}
\end{table*}

In a realistic experimental setup, the measurement outcomes from the photodiodes are noisy voltage measurements, rather than photon numbers. Hence, there are additional considerations from the measurement devices which we summarize in Appendix~\ref{appendix: Practical-Implementation} in order to prevent overestimating the conditional min-entropy. This gives a realistic SDI protocol for practical implementation. Importantly, we extend the SDI protocol to accommodate any optical beam splitter with arbitrary reflectivity $r_0$ to generate certified randomness. By relaxing this assumption in Ref.~\cite{drahi2020certified}, the extended SDI protocol becomes more robust as it not only captures the experiment setup realistically, but also allows for the usage of only fully passive optical elements.

As such, our extended SDI protocol takes into account that the $H_{\min,r_0}^{\text{SDI}}(X|E)$ varies with different values of $r_0$. Similar to~\cite{drahi2020certified}, we can consider the worst-case scenario in which Eve always inputs her optimal state $\hat{\rho}_E=\ket{n}\bra{n}$, where $\ket{n}$ is the Fock state. This optimal input state maximizes her guessing probability of $x$, which remains true even if we consider a general attack in which the photons are entangled for all $m$ measurement rounds. 

To determine $H_{\min,r_0}^{\text{SDI}}(X|E)$, in Appendix~\ref{appendix: Proof-of-Extended-SDI-Protocol}, we show that the outcome of $x$ could be effectively modelled by a binomial distribution, where the photons go to photodetector A with a probability of $r_0$. Then, Eve's best guessing probability, denoted by $p_{\text{guess}}$, occurs precisely at the peak (mean value) of $x$. To further maximize her $p_{\text{guess}}$, Eve ensures that exactly $n_R^-$ number of photons enter the randomness generation measurement as $p_{\text{guess}}$ decreases with increasing values of $n_R$. This gives a lower bound to $\kappa$ in Definition~\ref{def: Definition 1}. 

To evaluate this lower bound, the mean value of $x$ has to be determined. Since the product of $r_0 n_R^-$ is not always an integer, rounding to the correct integer is required to obtain the maximal $p_{\text{guess}}$ for a binomial distribution. Thus, this peak value of $x$ occurs exactly at $\mu_x = 2 \lceil(n_R^-+1)r_0-1\rceil -n_R^-$ (Appendix~\ref{appendix: Proof-of-Extended-SDI-Protocol}). By further taking into account the width of the voltage bin and the ENOB of the ADC for a practical implementation (Appendix~\ref{appendix: Practical-Implementation}), the effective range of $x\in\mathcal{X}^{\text{SDI}}_{r_0}$ can be obtained to estimate $\kappa$. The complete proof for the extended SDI protocol can be found in Appendix~\ref{appendix: Proof-of-Extended-SDI-Protocol}. As for the failure probability $\epsilon_{\text{fail}}$, from Ref.~\cite{drahi2020certified}, it can summarized as the following: (1) $\epsilon_{\text{fail}}$ is the security parameter for $m=1$ round of measurement when $n_R\notin [n_R^-,n_R^+]$ even if the certification test $\mathcal{P}$ is passed, (2) $\epsilon_-$ ($\epsilon_+$) is the security parameter when $n_R<n_R^-$ ($n_R>n_R^+$), (3) $\epsilon_{\gamma_C}$ is the security parameter when the electronic noise of photodetector C is larger than a desired upper bound $\Tilde{\gamma}_C$, i.e. $\abs{\gamma_C} > \Tilde{\gamma_C}$. With the above established, we formally present the extended SDI protocol in Table.~\ref{tab: Extended SDI protocol}. Lastly, the randomness generation rate of the QRG is given by $R_G = R_{\text{sample}}\times \kappa/b$, where $\kappa$ is the min-entropy per sample and $R_\text{sample}$ is the acquisition speed of the ADC. 

\subsection{Randomness Extraction}\label{sec: Randomness Extraction}
To obtain uniform random bits from the raw quantum data, randomness extraction is performed using a two-universal hash function, which ensures that the output is statistically close to uniform, even in the presence of potential (classcial or quantum) side information accessible to an adversary. More specifically, the Toeplitz randomness extractor is used due to its simplicity in its implementation on Field-Programmable Gate Array (FPGA) \cite{Herrero_Collantes_2017,zhang2016fpga,zheng20196}. The Toeplitz randomness extractor is made up of a matrix with block size $l\times h$, where $l$ is the number of bits extracted from the raw random bits of length $h$ from the QRG. Using the randomness extraction definition in Ref.~\cite{drahi2020certified,tomamichel2011leftover}, we specify it for the Toeplitz randomness extractor in the following theorem.

\begin{theorem}
    A certified SDI ($m$, $\kappa$, $\epsilon_{\text{fail},m},\epsilon_C$)-randomness generation protocol can be processed with a random seed of length $h+l-1$, where $h=mb$ and $b$ is the ADC's bit resolution, via Toeplitz randomness extractor to produce a certified SDI random string of length $l$ given by
    \begin{eqnarray}
        l = \kappa +2 -\log_2\frac{1}{\epsilon_{\textnormal{hash}}^2}
    \end{eqnarray}
    that is $\epsilon_C$ complete and $\epsilon_l$ secure, where $\epsilon_l = \epsilon_{\textnormal{hash}}+\epsilon_{\textnormal{fail},m}$ and $\epsilon_{\textnormal{hash}}$ is the security parameter for the randomness extraction.
\end{theorem}

The compression ratio is given by $r=l/h$, and a higher $r$ means that a greater amount of randomness could be extracted from the raw bits. Given that both the extended SDI protocol security and randomness extraction are composable, to produce a string of random numbers of length $L$ that concatenates $l$ bits of random numbers $t$ number of times, i.e. $L=t\times l$, the overall security parameter $\epsilon$ of the SDI-QRNG is \cite{drahi2020certified}
\begin{eqnarray}
    \epsilon = t\epsilon_l \geq t(\epsilon_{\text{hash}}+m\epsilon_{\text{fail}}).
\end{eqnarray}

The total bit rate of SDI-QRNG depends on either the sampling rate of the ADC, $R_{\text{sample}}$, or the clock speed of the FPGA, $R_{\text{hash}}$, where the slower factor becomes the bottleneck. The random number generation rate of the SDI-QRNG is $R_S=\min\{R_{\text{sample}},R_{\text{hash}}\}\times r$. Finally, from the \textit{Completeness} of the protocol, the average random number generation rate is $\left<R\right> = (1-\epsilon_C)\times R_S$.

\section{Experiment Setup}\label{sec:Experiment}

The experimental setup for Fig.~\ref{fig: SDI-QRNG Schematic} consists of the following. First, the untrusted source is a laser source (Koheron LD101) operating at $\lambda = 1550 \text{nm}$, with a typical linewidth of $5$MHz. A single photodetector (Koheron PD100-DC) is used for the certification measurement (PD-C in Fig.~\ref{fig: SDI-QRNG Schematic}). For the randomness generation (PD-A and PD-B in Fig.~\ref{fig: SDI-QRNG Schematic}), we used a pair of balanced photodetectors (Koheron PD100B-AC) with a Common Mode Rejection Ratio (CMRR) of $35$dB at $1$MHz. For the purpose of our demonstration, we have assumed that the responsivity for both balanced photodetectors to be identical. The technical specifications of these photodetectors are shown in Table~\ref{tab: PD details}. Upon characterization, the reflectivity for the optical beam splitter for certification measurement is $r_1 = 0.109$ (Thorlabs TN1550R2A2) and the randomness generation measurement has a fixed fiber beam splitter of $r_0 = 0.513$ (Thorlabs TN1550R5A2). The device used to sample and post-process the measurements is the Red Pitaya STEMlab 125-14. It comes with an FPGA (Xilinx Zynq 7010), where its clock rate is $R_{\text{hash}}=125$MHz. This board also has an ADC with $b=14$ bit resolution (LTC2145-14) and an ENOB of $\Delta_{\text{ADC}} = 11.83$ bits. The voltage range is $\pm 1$V in the low voltage setting and has a sampling rate of $R_{\text{sample}}=125$MS/s.

\begin{table}[h]
    \centering
    \caption{Technical information for the photodetectors. PD: photodetector}
    \label{tab: PD details}
    \begin{ruledtabular}
    \begin{tabular}{ccc}
    {Paramters}& {Certification PD} & {Balanced PD} \\ \midrule
    Bandwidth ($\text{BW}$) & 110MHz & 100MHz\\
    Transimpedance Gain ($\text{G}$) & 3.9k$\Omega$ &  39k$\Omega$ \\ 
    Responsivity ($\eta$) & 1.03A/W~ 
    & 0.9A/W\\
    Saturating Optical Power & 0.6mW & 1.5mW/PD
    \\
    \end{tabular}
    \end{ruledtabular}
\end{table}

\section{Results} \label{sec:Results}

\subsection{Extended SDI Protocol Analysis}\label{sec: Extended SDI Protocol Analysis}

To evaluate the expected certified randomness from the extended SDI protocol for our set-up, we analyze $H_{\min,r_0}^{\text{SDI}}(X|E)$ with different optical powers and $r_0$ using the measured photon number $n_c$ of the certification measurement. For simplicity, we choose $\epsilon_{\gamma_C} = \epsilon_- = \epsilon_{\text{fail}}/2$, while ensuring that $\epsilon_->\epsilon_+$ for all $r_0$ used in this analysis. From the measurement result of $n_C^-$, the corresponding $n_R^-$ can be obtained from $\epsilon_-$ in Eq.~\ref{eq: epsilon- and epsilon+}. On the other hand, we set $n_R^+$ as the number of saturating photons of the photodetectors at the randomness generation measurement. Utilizing these values, $\epsilon_+$ can be obtained via $n_C^+$. The result of $H_{\min,r_0}^{\text{SDI}}(X|E)$ for $m=1$ round of measurement is numerically computed and shown in Fig.~\ref{fig: SDI Hmin plot} with a fixed security parameter of $\epsilon_{\text{fail}}=10^{-20}$. We allow almost all samples to pass the certification test $\mathcal{P}$ by setting $\epsilon_C=10^{-6}$, and use Eq.~\ref{eq: explicit cert test P} in Appendix~\ref{appendix: Completeness-Derivation} to determine the voltage limit for the certification test $\mathcal{P}$. 

From Fig.~\ref{fig: SDI Hmin plot}, we observe that the randomness of the extended SDI protocol decreases when $r_0$ increases from $0.5$. When $r_0=1$, as expected, no randomness can be derived as all photons will reach only one photodetector deterministically. It is interesting to note that $H_{\min,r_0}^{\text{SDI}}(X|E)$ does not drop drastically as $r_0$ increases, where around $75\%$ of the maximum randomness is still present even for $r_0=0.9$ at $2.00$mW of input optical power. The small inset graph in Fig.~\ref{fig: SDI Hmin plot} focuses on the relationship of $H_{\min,r_0}^{\text{SDI}}(X|E)$ and $r_0$. In general, as long as $r_0$ is not too close to the extremes of $0$ or $1$, the extended SDI protocol can still generate certified randomness.

For every $r_0$, importantly, the randomness drops to $0$ when one of the photodetectors at the randomness generation measurement is saturated, i.e.~when $n_R$ is more than $n_R^+$. We note that as the optical power increases, it could also lead to the situation where $\epsilon_+$ surpasses $\epsilon_-$~\cite{drahi2020certified}. In this case, the security of $\epsilon_{\text{fail}}$ will no longer hold, leading to an $H_{\min,r_0}^{\text{SDI}}(X|E)$ of $0$. On the other hand, in the regime of low optical power, no randomness is initially produced because the electronic noise of the certification photodetector is still significantly larger compared to the number of photons impinging onto the photodetector. As a result, no positive value for $n_C^-$ can be obtained to generate certified randomness.

\begin{figure}[t!]
    \centering
    \includegraphics[width=0.9\linewidth]{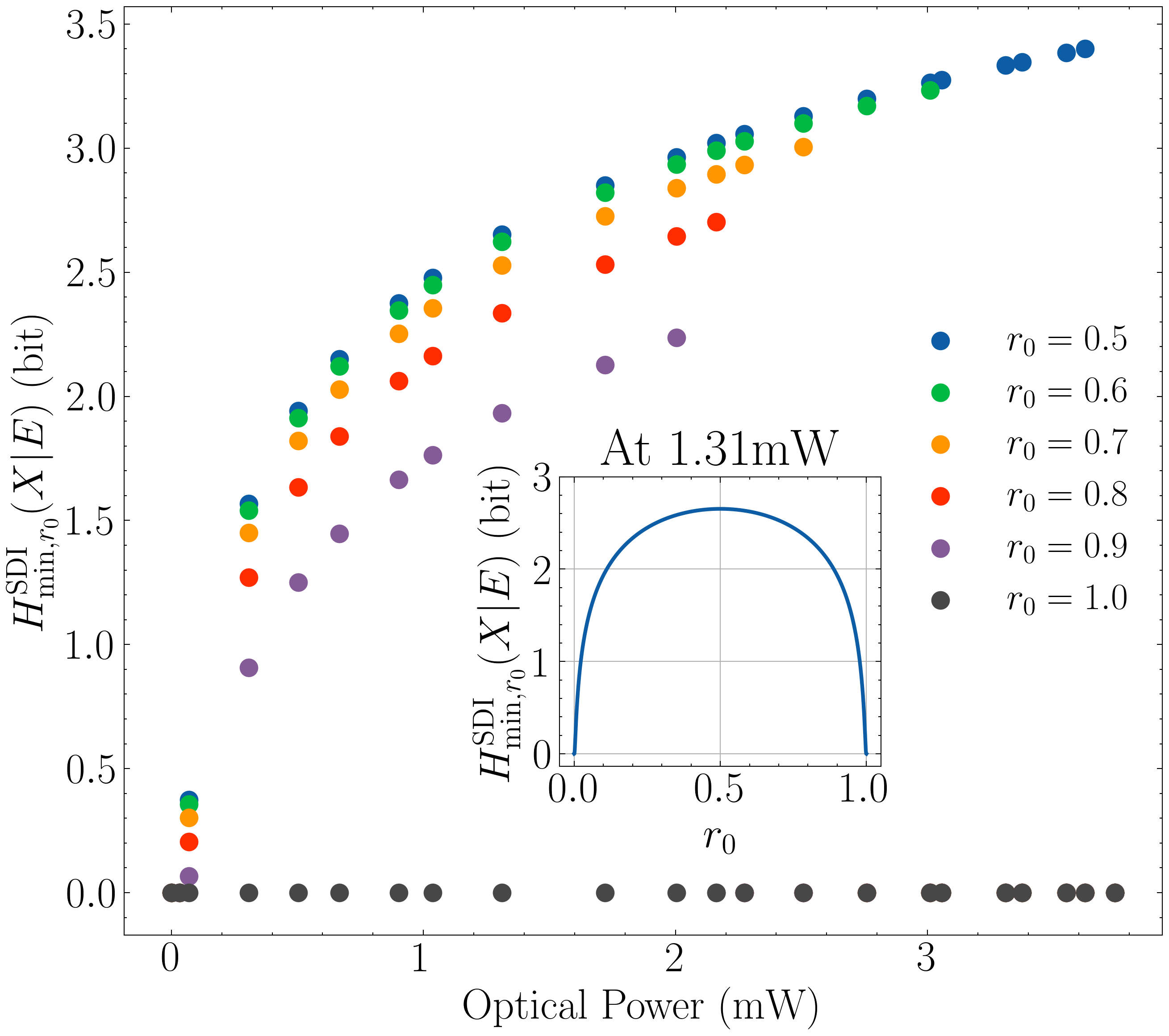}
    \caption{$H_{\min,r_0}^{\text{SDI}}(X|E)$ for $m=1$ measurement is computed with $r_0\in[0.5,1]$ at a fixed $\epsilon_{\text{fail}} = 10^{-20}$. The inset figure presents the changes in $H_{\min,r_0}^{\text{SDI}}(X|E)$ with respect to $r_0\in[0,1]$ at a particular optical power input of $1.31$mW.}
    \label{fig: SDI Hmin plot}
\end{figure}

To understand the trade-off between the security and performance for our extended SDI protocol, we can compare it with an DD-QRNG protocol. The DD protocol used here is a homodyne protocol that generates randomness by measuring the amplitude quadrature of the vacuum signal with a strong coherent local oscillator \cite{gabriel2010generator,haw2015maximization,zheng20196}. For a proper comparison, we need to devise an unbalanced DD-QRNG protocol to accommodate the various beam splitting ratios $r_0$. By considering the unbalanced homodyne detection model in Ref.~\cite{PhysRevA.87.052309}, the variance of the difference measurement output in terms of the photon number is (Appendix~\ref{appendix: Math-Details-for-Unbalanced-Homodyne-Detection})
\begin{eqnarray}
    \sigma_{\text{UHD}}^2= \left[(2r_0-1) \overline{n}_{R}f\right]^2 + 4r_0(1-r_0)\overline{n}_{R} + \sigma_{n_D}^2
\end{eqnarray}
\noindent where $f=\sqrt{\text{Var}(\hat{n}_{R})}/\Bar{n}_{R}$ \cite{chi2011balanced,PhysRevA.98.012312} is the ratio of the fluctuation of the intensity to its mean photon number $\overline{n}_{R}$ (otherwise known as Relative Intensity Noise), $\text{Var}(\hat{n}_{R})=\langle\hat{n}_R^2\rangle_{\alpha_{n_R}} -\langle\hat{n}_R\rangle^2_{\alpha_{n_R}}$ is the variance evaluated over the coherent state $\ket{\alpha_{n_R}}$, and $\sigma_{n_D}^2$ is the electronic noise of the photodetectors in photon numbers. The first term captures the contribution of the fluctuation of the local oscillator due to the imperfect cancellation of the intensity at the unbalanced detection, whereas the second term is the vacuum fluctuation $\sigma_Q^2= 4r_0(1-r_0)\overline{n}_{R}$, which is the source of the quantum randomness, quantified by $H_{\min,r_0}^{\text{DD}}(X|E)$ (see Appendix~\ref{appendix: Math-Details-for-Unbalanced-Homodyne-Detection}). 

\begin{figure}[t!]
    \centering
    \includegraphics[width=0.9\linewidth]{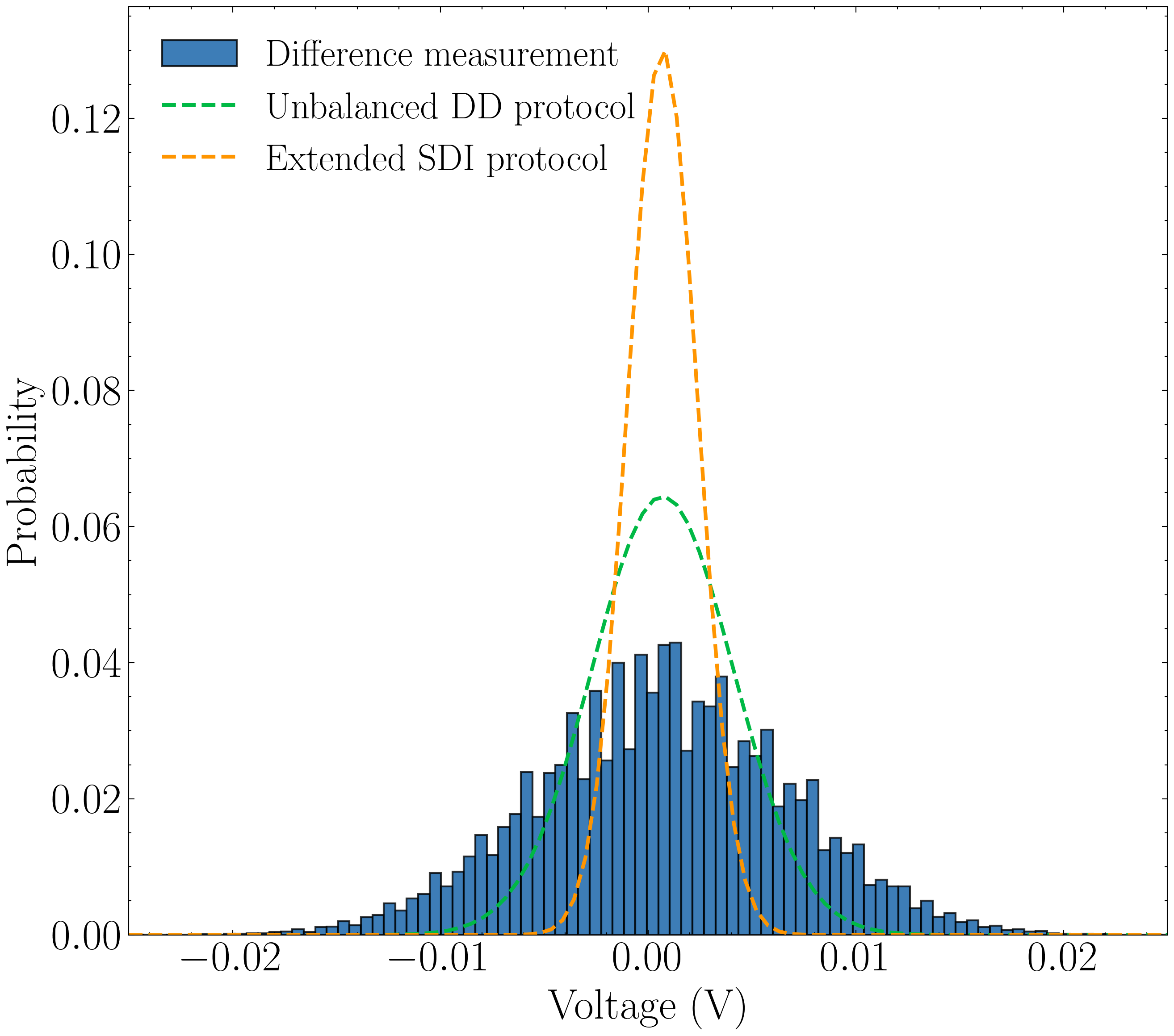}
    \caption{Probability distributions (with binwidth of ENOB) for acquired difference measurement (blue bin), computed unbalanced DD (green dotted-line) and extended SDI protocol (orange dotted-line) at $r_0=0.7$ with $2.57$mW of optical power.}
    \label{fig: comparison_between_DD_and_SDI}
\end{figure}

To illustrate the trust-performance trade-off over different randomness generation protocols, we compare the difference measurement distribution and the conditional distribution of the unbalanced DD and the extended SDI protocol at $r_0=0.7$, as shown in Fig.~\ref{fig: comparison_between_DD_and_SDI}. Here, we use a tunable fiber beam splitter to achieve $r_0=0.7$. The rest of the experimental parameters are the same as our extended SDI protocol analysis. 

The difference measurement acquired from the balanced photodetectors consists of all noise parameters in $\sigma_{\text{UHD}}^2$. Meanwhile, the probability distribution of the vacuum noise, $\sigma_Q^2$, is computed with $\overline{n}_R$ using $\overline{n}_C$ from the certification photodetector and is represented as the green dotted line. This corresponds to $H_{\text{min}}^{\text{DD}}(X|E) = 3.957$ bits per sample. The difference between these two distributions further demonstrates the presence of local oscillator fluctuations resulting from unbalanced detection, in addition to the electronic noise must be taken into account to avoid overestimating the randomness. 

The distribution of the extended SDI protocol is computed with $\sigma_A^2=r_0(1-r_0)n_R^-$ (see Appendix~\ref{appendix: Proof-of-Extended-SDI-Protocol}) and is represented by the orange dotted line, corresponding to $H_{\text{min}}^{\text{SDI}}(X|E) =2.946$ bits per sample. The randomness of the extended SDI protocol differs from the unbalanced DD protocol by $1.011$ bit per sample at $2.57$mW of optical power, which is a $25.54\%$ decrease in randomness. In fact, as shown in Appendix~\ref{appendix: HminDD-HminSDI}, their difference tends towards $1$ bit of randomness in the asymptotic limit of large $\overline{n}_R$ and $n_R^-$. Moreover, in this regime, it will converge to $1$ bit of randomness even as their bit depth increases. This suggests that one can opt for a higher ENOB to minimize the performance trade-off when switching from a DD model to an SDI model.

\subsection{Real-time SDI-QRNG Performance}\label{sec: Real-Time SDI-QRNG Performance}

We evaluate the online performance of our extended SDI-QRNG protocol by generating random numbers in a real-time operation. To this end, we employ PYNQ~\cite{PYNQ}, an open source project from Xilinx that facilitates the deployment of FPGA images and acquires their output in the Python environment. As Red Pitaya does not support the functionality of PYNQ natively, an operating system containing PYNQ is installed from an open-source code~\cite{RP-PYNQ}. As such, we further performed the necessary calibration for the Red Pitaya acquisition functions. 

As we aim to demonstrate real-time operation using a cost-effective FPGA-based system that handles both acquisition and post-processing, it is crucial to optimize resources to maximize the random number generation rate. Understanding the resource consumption, along with the hashing security parameter $\epsilon_{\text{hash}}$, across different hashing block sizes, is essential for selecting optimal FPGA parameters for this operation~\cite{zheng20196}.
We provide further details of our implementation and optimization in Appendix~\ref{appendix: FPGA-Architecture-and-Analysis}.

We operate our set-up in a real-time manner using an optimal optical power input of $3.43 \text{mW}$, where a fixed optical beam splitter of $r_0=0.513$ is used. The other relevant parameters used are presented in Table~\ref{tab: parameters for SDI-QRNG}. From this performance evaluation, the QRG generation rate is $R_G  = 0.419$Gb/s. In our case, the bottleneck of our SDI-QRNG is the ADC acquisition rate $R_{\text{sample}}$, hence the random number generation rate is $R_S = R_{\text{sample}}\times r = 0.350$Gb/s, with a compression ratio $r=19.98\%$ (for min-entropy per sample of $3.354$ bits over $14$ bits with $\epsilon_{\text{hash}} = 2.33\times 10^{-16}$). Finally, the average QRNG throughput is $\langle R \rangle = 0.347$Gb/s with $1-\epsilon_C = 0.992$ and an overall composable security of $\epsilon=8.12\times10^{-13}$. The NIST statistical test suite (STS) for random number generators~\cite{NIST_STS,FAST-NIST} is conducted using an accumulated 1 Gbit of random bits, and the test is successfully passed (see Appendix~\ref{appendix: NIST-Test-Results}).

\begin{table}[t]
    \caption{Parameters for the real-time SDI-QRNG operation.
    }
    \centering
    \begin{ruledtabular}
    \begin{tabular*}{\linewidth}{@{\extracolsep{\fill}}ccc}
    Parameters & Notations & Value \\
    \midrule
        Reflectivity & $r_0$ & $0.513$ \\
        Min-entropy per sample & $H_{\min,r_0}^{\text{SDI}}$ & $3.354$ bits \\
        Hash cycles performed & $t$ & $2500$ \\
        Samples per hash & $m$ & $183$\\
        Length of hashing input & $h$ & $2562$ bits\\
        Length of hashing output & $l$ & $512$ bits \\
        Compression ratio & $r$ & $19.98\%$ \\
        Sample failure prob & $\epsilon_{\text{fail}}$ & $5.00 \times 10^{-19}$\\
        Hashing failure prob & $\epsilon_{\text{hash}}$ & $2.33 \times 10^{-16}$ \\
        Single hashing failure prob & $\epsilon_l$ & $3.25 \times 10^{-16}$\\
        Total failure prob & $\epsilon$ & $8.12 \times 10^{-13}$\\
        Certification failure prob & $\epsilon_C$ & $0.008$ \\
        Randomness generation rate & $R_G$ & $0.419$ Gb/s\\
        Randomness extraction rate & $R_S$ & $0.350$ Gb/s\\
        Average bit rate & $\langle R \rangle$ & $0.347$ Gb/s\\
        $\epsilon$-random bits per string & $L$ & $1.28$ Mb\\
    \end{tabular*}
    \end{ruledtabular}
    \label{tab: parameters for SDI-QRNG}
\end{table}

\subsection{Experimental Verification of SDI Protocol Implementation}\label{sec: Verification of Implementation}

\begin{figure*}[t!]
    \centering
    \includegraphics[width=0.8\linewidth]{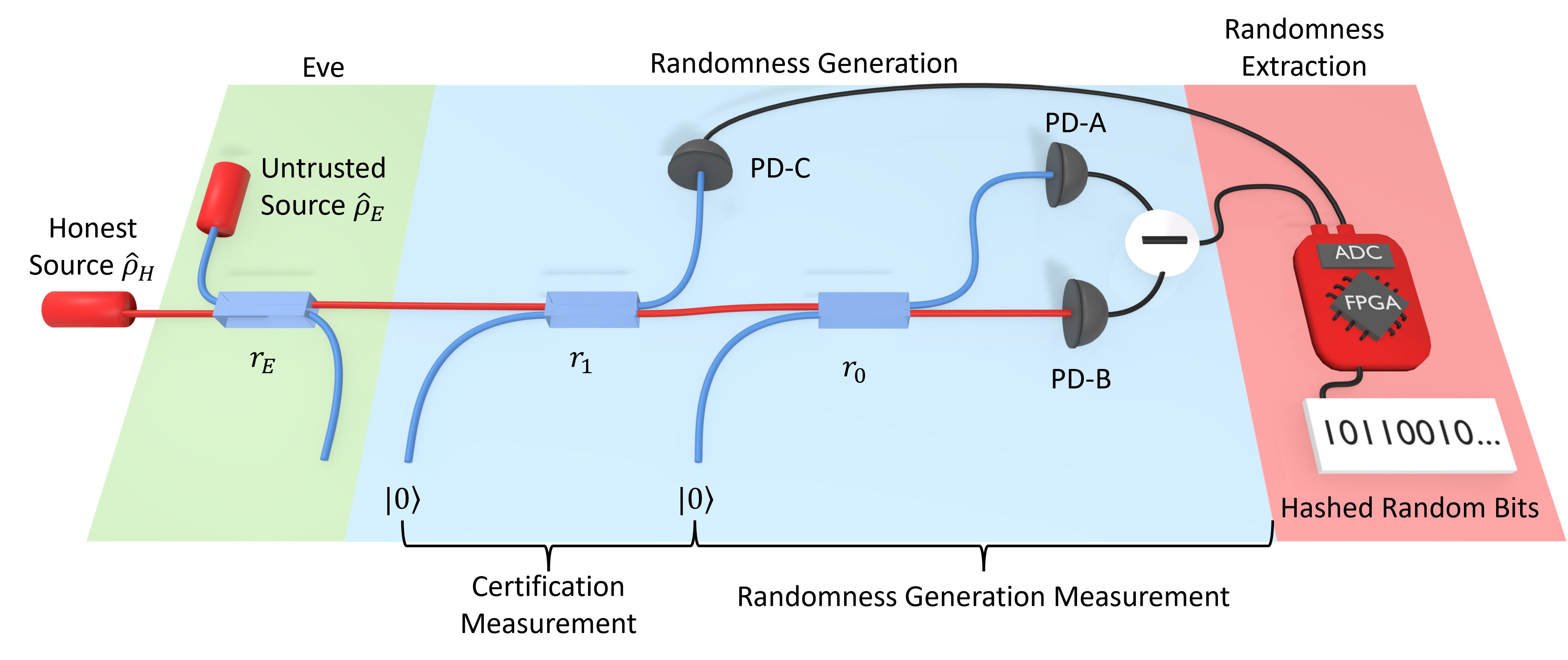}
    \caption{Schematic of Eve's light intensity operation. A fiber beam splitter of reflectivity $r_E$ is inserted between the honest source $\hat{\rho}_H$ and the measurement devices so that Eve can input her light source $\hat{\rho}_E$ into the QRNG. PD: photodetector}    
    \label{fig: light intensity verification schematic}
\end{figure*}

\begin{figure}[t!]
    \centering
    \includegraphics[width=0.9\linewidth]{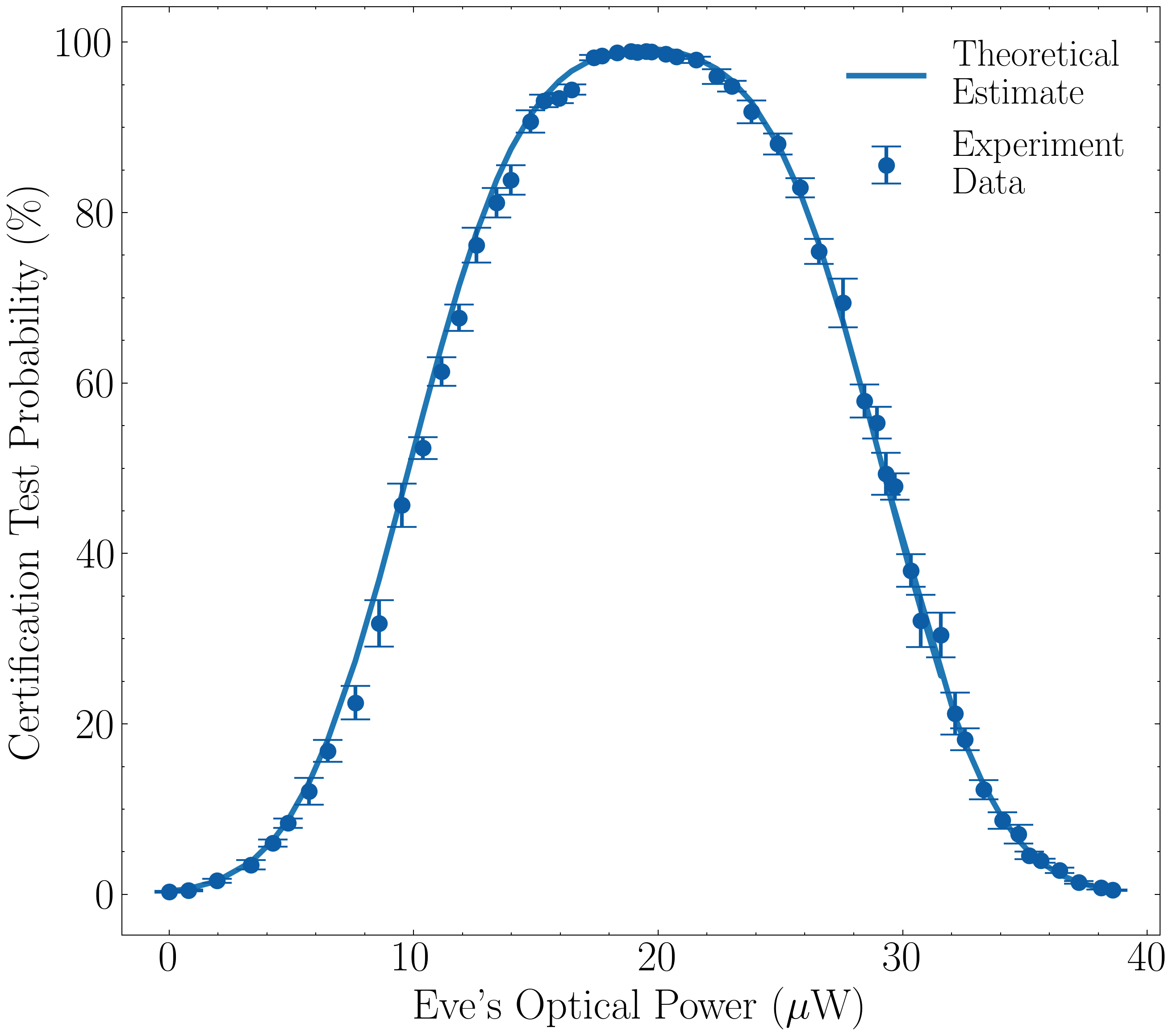}
    \caption{Experimental results for the light intensity verification. The blue experimental data points are the passing probability of the certification test, along with its error bar. The theoretical estimate is plotted in blue line.}
    \label{fig: light intensity verification result}
\end{figure}

To further evaluate our experimental SDI protocol set-up over an untrusted light source, we emulate the scenario whereby an eavesdropper could inject and manipulate additional light sources into the QRNG. We start by assuming an honest light source, $\hat{\rho}_{\text{H}}$, entering the QRNG. The eavesdropper, Eve, is allowed to change the total light intensity by injecting her own coherent light source (Koheron LD101), $\hat{\rho}_{\text{E}}$, by placing an additional beam splitter between $\hat{\rho}_{\text{H}}$ and the measurement devices, as illustrated in Fig.~\ref{fig: light intensity verification schematic}. Here, we set the additional beam splitter with reflectivity $r_E = 0.0105$, which allows fine-tuning of the injected light entering the QRNG. Since the certification test is sensitive to intensity changes, the full spectrum of the passing probability response can be obtained. 

To illustrate the impact of Eve's malicious activity over the light source, we set the protocol to have $1-\epsilon_C = 99.2\%$ passing probability. To compensate for the addition of $r_E$ from the beam splitter, the input power from $\hat{\rho}_H$ is initially adjusted so that $0.5\%$ of the samples passes the test in the absence of Eve's light source. We see in Fig.~\ref{fig: light intensity verification result} that when Eve injects $18.9\mu\text{W}$ of optical power into the system, an optimal passing probability of $99.2\%$ is reached. As Eve adjusts her optical power away from this optimal point, the passing probability decreases, demonstrating the security feature of the SDI protocol when the light intensity varies around the optimal input for the certification test $\mathcal{P}$. The theoretical estimations for this probability, derived from Eq.~\ref{eq: completeness of coherent state} (and Eq.~\ref{eq: explicit cert test P}), optimized for a coherent light source of $1550$nm, exhibit a strong correspondence with the experimental data, as evidenced by an R-square value of $0.9978$ in Fig.~\ref{fig: light intensity verification result}. This illustrates that our verification model is robust and validates the protocol's response to intensity variations via the certification measurement. In other words, as the intensity deviates from the optimal values, the average certified randomness generation rate $\langle{R}\rangle$ will be scaled down according to its corresponding passing probability $1-\epsilon_C$ (Sec.~\ref{sec: Randomness Extraction}).

\section{Discussion and Conclusion} \label{sec:Discussion and Conclusion}

When comparing the performance of the unbalanced DD protocol with our extended SDI protocol (in Sec.~\ref{sec: Extended SDI Protocol Analysis}), several key advantages emerge. First, our protocol is notably easier to implement, even when the light source is entirely untrusted. In particular, by explicitly accounting for an unbalanced beam splitter at the difference measurement process, our model becomes inherently robust against local oscillator fluctuations—an issue that must be addressed explicitly in the unbalanced DD protocol. Furthermore, as discussed in Sec.~\ref{sec: Extended SDI Protocol Analysis}, the performance trade-off of the extended SDI protocol, as compared to the DD approach, can be optimized by resorting to an ADC with higher bit depth.

Secondly, by removing the requirement for perfectly balanced photodetectors, our protocol widens the technological applicability of certifiable QRNGs. This is relevant for fiber-based QRNG systems that utilize ultra-high-speed balanced detectors \cite{struszewski2017characterization,FraunhoferBPD} to achieve high-bit rate. Maintaining a high level of optical field cancellation (i.e. high CMRR) in such systems is notoriously difficult at high bandwidths and typically demands finely tuned optical path lengths \cite{thorlabs_pdb482c}. By translating minimal guaranteed CMRR into an equivalent beam splitting ratio \cite{chi2011balanced,huang2020practical}, our protocol enables a conservative, yet secure, estimation of certified randomness under realistic constraints. 

Thirdly, even in terms of photonic integrated circuits (PIC) QRNG systems based on balanced detection, which target device miniaturization for wider applications, 
our protocol presents an opportunity to minimize both the footprint of the system and the complexity of implementation while achieving light source independence. For instance, during the detection stage in PIC, there could be differences in the photodiode efficiencies, as well as finite on-chip electronic subtraction, which lead to imbalance detection, or a non-negligible CMRR \cite{Bruynsteen:21,wang2024compact,bai202118,raffaelli2018homodyne}. Recent progress in PIC for SDI-QRNG \cite{li2024chip,bertapelle2023high,du2023source} highlights the feasibility of implementing our protocol within these compact and scalable architectures.

Lastly, the SDI protocol allows, in principle, \textit{any} light source to operate the QRNG without requiring a new security proof; the only change required is a simple update on the certification test $\mathcal{P}$. For example, an incoherent, broadband amplified spontaneous emission (ASE) source can serve as an honest implementation, provided that an optical filter is employed before the measurement devices, to ensure that the wavelength of the laser entering has a narrow linewidth centered at $1550$nm \footnote{In general, to  ensure security against arbitrary untrusted light source, regardless of its wavelength or intrinsic properties, it is desirable to integrate a narrow $1550$nm optical filter. This measure will ensure that the light source is effectively filtered prior to entry, fulfilling the assumptions of the SDI protocol.}~\cite{williams2010fast,qi2017true,yang2020randomness}. This configuration satisfies both the assumptions of the SDI protocol and the requirements of the measurement devices. We present a detailed theoretical treatment of the ASE source and its characterization process to determine the certification test $\mathcal{P}$ in Appendix~\ref{appendix: ASE-Based-SDI-QRNG-Theory}. 

In conclusion, we implemented a practical SDI-QRNG using compact and readily available components, demonstrating the feasibility of a lightweight and cost-effective QRNG. Our prototype generates random numbers at an average rate of $0.347$ Gb/s with an overall composable security of $\epsilon = 8.12\times 10^{-13}$. Furthermore, our Red Pitaya–based implementation—featuring a single board integrating both ADC and FPGA—can serve as a versatile platform for other QRNG architectures. To the best of our knowledge, this is the first demonstration of a randomness extractor implemented on a single-board solution of this kind. Our system reliably generates certified randomness across a broad range of beam splitter ratios, thereby simplifying experimental implementation and reducing dependence on idealized measurement assumptions. To validate the security of our protocol, we experimentally performed adversarial manipulation of the light source, with results aligning closely with theoretical predictions. These findings establish our system as a robust, high-performance, and fully passive SDI-QRNG, well-suited for quantum-safe applications such as quantum key distribution and post-quantum cryptography.

\section*{Acknowledgements} \label{sec:acknowledgements}
We thank Nathan Walk for helpful discussions. We acknowledge funding support from National Research Foundation, Singapore and A*STAR under its Quantum Engineering Programme (National Quantum-Safe Network, NRF2021-QEP2-04-P01), the start-up grant for Nanyang Assistant Professorship of Nanyang Technological University, Singapore, and the Tier 1 MOE grant RT1/23 ``Catalyzing quantum security: bridging between theory and practice in quantum communication protocols". 

\appendix
\onecolumngrid
\newpage

\section{Practical Implementation Voltage POVM}\label{appendix: Practical-Implementation}

\begin{figure}[h]
    \centering
    \includegraphics[width=0.7\linewidth]{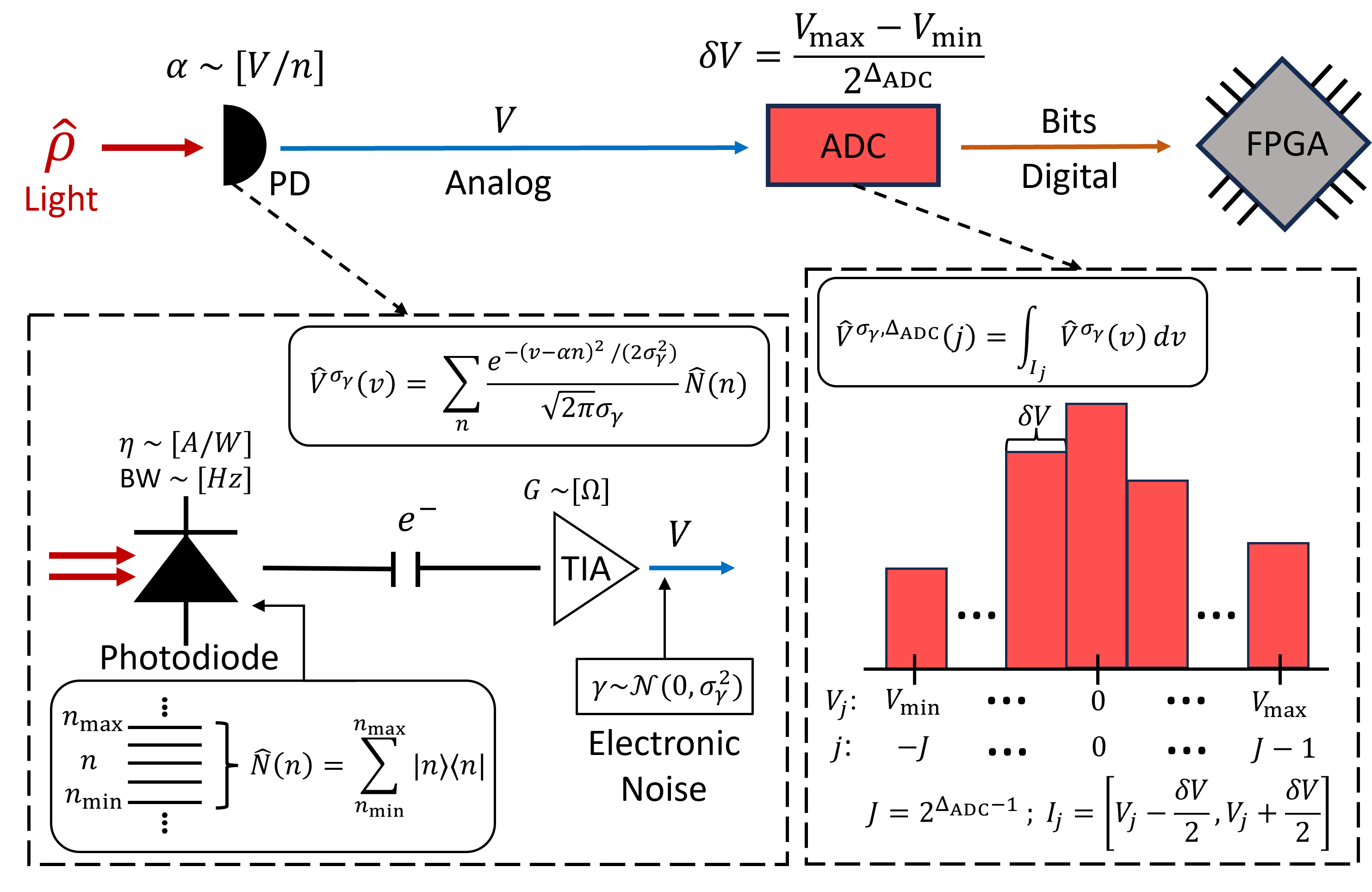}
    \caption{Practical implementation of the photodetectors and ADC components along with their POVM elements.}
    \label{fig: PD overview}
\end{figure}

According to \cite{drahi2020certified}, there are three factors that should be taken into account when estimating the conditional min-entropy for a practical QRNG, as illustrated in Fig.~\ref{fig: PD overview}. Firstly, when light is incident onto the photodiode with a finite operating photon number range $[n_{\min},n_{\max}]$, the light is converted into photocurrent, indicated by $e^-$ in Fig~.\ref{fig: PD overview}. Subsequently, this photocurrent will be converted for voltage measurement by the Transimpedance Amplifier (TIA). The conversion factor for this process can be represented by~\cite{drahi2020certified},
\begin{eqnarray}
    \alpha = \frac{hc\text{BW}\eta G}{\lambda}
\end{eqnarray}
where $h$ is the Planck's constant, $c$ is the speed of light, BW is the bandwidth of the photodetector, $\eta$ is the responsitivity of the photodetector at a particular wavelength $\lambda$, and $G$ is the gain of the TIA. The $\alpha$ is specific to the chosen $\lambda$. Hence, to ensure that $\alpha$ is constant throughout the acquisition window, either the linewidth of the laser used must be narrow or an optical filter is employed.

Secondly, the photodetector exhibits intrinsic electronic noises, which are classical and Gaussian distributed with a noise variable of $\gamma$ and a variance of $\sigma_{\gamma}^2$. This causes voltage measurements to be noisy, and for $n$ number of photons, the resultant voltage measurement is given by $v=\alpha n + \gamma$. The Positive Operator-Valued Measure (POVM) element of voltage $\{\hat{V}^{\sigma_{\gamma}}(v)\}$, with a general photon number measurement projector $\hat{N}(n)=\ket{n}\bra{n}$, is given by
\begin{equation}
    \hat{V}^{\sigma_{\gamma}}(v) = \sum_{n = n_{\min}}^{n_{\max}} \frac{e^{-(v-\alpha n)^2/(2\sigma_{\gamma}^2)}}{\sqrt{2\pi}\sigma_{\gamma}}\hat{N}(n).
\end{equation}

Lastly, the analog voltage signals from the photodetector are fed into the Analog-to-Digital Converter (ADC), which has a finite voltage range $[V_{\min},V_{\max}]$ and a resolution of $b$ bits that outputs $2^b$ number of voltage bins. The ADC also exhibits internal electronic noise, and the Effective Number of Bits (ENOB) of the ADC, $\Delta_{\text{ADC}}$, must be considered during the estimation of the conditional min-entropy. Thus, the resultant number of voltage outputs of the ADC reduces to $2^{\Delta_{\text{ADC}}}$. This results in every $j$-th voltage bin having a width of $\delta V = (V_{\max}-V_{\min})/2^{\Delta_{\text{ADC}}}$. The voltage measurement of a particular $j$-th bin is given by the integral over the interval $I_j$. Combining all these factors, the realistic voltage measurement is represented with the following POVM elements:
\begin{equation}\label{eq: overall POVM voltage}
    \hat{V}^{\sigma_{\gamma},\Delta_{\text{ADC}}}(j) = \int_{I_j}\hat{V}^{\sigma_{\gamma}}(v)dv,
\end{equation}
where the integration limit for $j$-th voltage bin is
\begin{eqnarray}
    I_j  = \left[V_j-\frac{\delta V}{2}, V_j + \frac{\delta V}{2}\right] \quad \text{s.t} \quad
     j = \mathbb{N} \cap \left[-2^{\Delta_{\text{ADC}}-1},2^{\Delta_{\text{ADC}}-1}-1\right].
\end{eqnarray}

\section{Proof of the Extended SDI Protocol} \label{appendix: Proof-of-Extended-SDI-Protocol}

The proof of our Extended SDI-QRNG protocol focuses on the derivation of the lower bound of the conditional min-entropy $\kappa$ for arbitrary $r_0$. Firstly, the POVM element of the measurement outcome of a general photon number $n_1$ and $n_2$ after an arbitrary beam splitter with reflectivity $r_{i\in[0,1]}$ is \cite{drahi2020certified}
\begin{eqnarray}
    \hat{M}\left(n_1, n_2\right)=&\frac{r_{i}^{n_1}\left(1-r_{i}\right)^{n_2}\left(n_1+n_2\right) !}{n_{1} ! n_{2} !}\times \left|n_1+n_2\right\rangle\left\langle n_1+n_2\right|
\end{eqnarray}

\noindent and $n = n_1+n_2$. From this, the POVM element for the difference measurement $\mathbb{X} = \left\{\hat{X}_{r_0}(x)\right\}$ with arbitrary $r_0$ is
\begin{eqnarray}
    \hat{X}_{r_0}(x) &&= \sum^{\left\lfloor\frac{n_R^++x}{2}\right\rfloor}_{n_A=\left\lfloor\frac{n_R^-+x}{2}\right\rfloor} r_0^{n_A}(1-r_0)^{n_A-x}\left(\begin{array}{c}
    {2n_A-x} \\ 
    {n_A}
    \end{array}\right)\nonumber\ket{2n_A-x}\bra{2n_A-x}\nonumber\\
    && = \sum^{n_R^+}_{n_R=n_R^-} r_0^{\left\lfloor\frac{n_R+x}{2}\right\rfloor}(1-r_0)^{\left\lceil\frac{n_R-x}{2}\right\rceil}\left(\begin{array}{c}
    {n_R} \\
    {\left\lfloor\frac{n_R+x}{2}\right\rfloor}
    \end{array}\right)\ket{n_R}\bra{n_R}
\end{eqnarray}
where the subscript A and B represent photodetector A and B at the randomness generation measurement in Fig.~\ref{fig: SDI-QRNG Schematic} respectively, $n_A = \lfloor(n_R+x)/2\rfloor$, $n_B = \lceil(n_R-x)/2\rceil$ \footnote{Since $n_A+n_B=n_R$, then by the property of the sum of the floor and ceiling function, $\lfloor(n_R+x)/2\rfloor + \lceil(n_R-x)/2\rceil = n_R$.}, and $\lfloor \cdot \rfloor$ ($\lceil \cdot \rceil$) is the floor (ceiling) function. In all of our analysis, the worst case in which Eve has complete knowledge of the photodetectors is always assumed. Hence, the overall electronic noise of the photodetectors at the difference measurement, denoted by $\gamma_D$, is given to Eve in a shot-by-shot basis. This implies that $\gamma_D$ can be effectively removed from the realistic POVM element of the difference measurement in Eq.~\eqref{eq: overall POVM voltage}, resulting in the following POVM element for estimating $\kappa$ given by
\begin{eqnarray}
    \hat{V}_D^{\Delta_{\text{ADC}}}(j) &&= \int_{I_j^D-\gamma_D} \hat{V}_D(v_D)dv_D \\ \nonumber &&= \sum_{x\in\mathcal{X}^{\text{SDI}}_{r_0}}\hat{X}_{r_0}(x)
\end{eqnarray}

\noindent for some range $\mathcal{X}^{\text{SDI}}_{r_0}= \{x : \alpha_D x +\gamma_D \in I_j^D\}$ and the subscript $D$ in the voltage POVM element represents the photodetectors at the difference measurement. To understand what the optimal photon state is that Eve can input into the QRG to achieve her best guessing probability, $p_{\text{guess}}$, we will need the following lemma in Ref.~\cite{drahi2020certified}.

\begin{lemma}[\textbf{Lemma 1 in \cite{drahi2020certified}}]
    For an $m$-round SDI protocol involving a measurement $\mathbb{Q} = \{\hat{Q}(q)\}$ in each round that is diagonal in the number state basis with POVM elements
    \begin{equation}
        \hat{Q}(q) = \sum_n c_n(q)\hat{N}(n), \quad \text{s.t.} \quad \sum_q \hat{Q}(q) = \mathbb{I},
    \end{equation}
    Eve's optimal strategy to maximize the probability of guessing a desired outcome $q'$ is to input a pure Fock state $\ket{n'}$ for each round. Moreover, this remains true for inputs with restricted support in the Fock state basis.
\end{lemma}
    
This lemma holds true even if we consider a general attack model where Eve chooses to input states that are entangled for all $m$ rounds \cite{drahi2020certified}. Given that our difference measurement POVM element, $\hat{X}_{r_0}(x)$, is diagonal in the number state basis, i.e. $\ket{n_R}\bra{n_R}$, and Eve's input state $\ket{n}$ has restricted support over the range $n\in[n_R^-,n_R^+]$, the condition for Lemma 1 is satisfied. Thus, for every round of measurement, Eve's best strategy to guess the outcome of $x$ is to input a pure Fock state $\ket{n}$ into the QRG and find her best $p_{\text{guess}}$, which occurs precisely at the peak of the probability distribution of $x$. 

The expectation value of $x$ is given by $\mu_x=\mu_A - \mu_B=r_0n_R - (1-r_0)n_R$. However, for a binomial distribution, the relevant values must be in integers and for any given $r_0\in[0,1]$, $\mu_A = r_0 n_R$ will not necessarily be an integer. This could cause an issue when it comes to rounding off $r_0 n_R$ to the nearest desired integer, as the probability of the binomial distribution might not always be maximal. Hence, by the property of binomial distribution for non-integers, there exists a positive integer $M$ such that $(n_R+1)r_0-1\leq M < (n_R+1)r_0$ always gives the maximal probability for $\mu_A=r_0 n_R$. This results in $\mu_x = 2M -n_R = 2\lceil(n_R+1)r_0-1\rceil - n_R$, which will always guarantee the maximal probability of $p_{\text{guess}}$. With this, $p_{\text{guess}}$ is expressed as 
\begin{eqnarray}
    p_{\text {guess }}&&=\max _{n \in\left[n_R^{-}, n_R^{+}\right]}\left\langle n\left|\sum_{x \in \mathcal{X}^{\text{SDI}}_{r_0}} \hat{X}_{r_0}(x)\right| n\right\rangle \nonumber \\
    && \leq \sum_{x \in \mathcal{X}^{\text{SDI}}_{r_0}} r_0^{\left\lfloor\frac{n_R^{-}+x}{2}\right\rfloor}(1-r_0)^{\left\lceil\frac{n_R^{-}-x}{2}\right\rceil}\left(\begin{array}{c}n_R^{-} \\ \left\lfloor\frac{n_R^{-}+x}{2}\right\rfloor\end{array}\right)
\end{eqnarray}

\noindent where in the last line, we use the following lemma to show that the inequality of $p_{\text{guess}}$ is due to the fact that the probability of the binomial distribution at $\mu_x$ decreases with increasing values of $n_R$. 

\begin{lemma}
    For any $0\leq r \leq 1$ and $n\in\mathbb{Z}^+$, the probability of the binomial distribution of the form
    \begin{eqnarray}
        P(n)= r^{\left\lfloor\frac{n+\mu}{2}\right\rfloor}(1-r)^{\left\lceil\frac{n-\mu}{2}\right\rceil}\left(\begin{array}{c}n \\ \left\lfloor\frac{n+\mu}{2}\right\rfloor\end{array}\right)
    \end{eqnarray}
    
    \noindent where it is maximal at its expectation value $\mu = 2M-n$, where $M \in \mathbb{Z}^+$ and $(n+1)r-1\leq M < (n+1)r$, and $P(n)$ decreases for increasing values of $n$.

\end{lemma}

\begin{proof}
    Consider the ratio of successive terms of $n$, where 
    \begin{eqnarray}
        \frac{P(n+1)}{P(n)}=\frac{r^{\left\lfloor\frac{n+1+\mu'}{2}\right\rfloor}(1-r)^{\left\lceil\frac{n+1-\mu'}{2}\right\rceil}\left(\begin{array}{c}n+1 \\ \left\lfloor\frac{n+1+\mu'}{2}\right\rfloor\end{array}\right)}{r^{\left\lfloor\frac{n+\mu}{2}\right\rfloor}(1-r)^{\left\lceil\frac{n-\mu}{2}\right\rceil}\left(\begin{array}{c}n \\ \left\lfloor\frac{n+\mu}{2}\right\rfloor\end{array}\right)} = 
        \frac{r^{M'}(1-r)^{ n+1 - M'}\left(\begin{array}{c}n+1 \\ M'\end{array}\right)}{r^{M}(1-r)^{n-M}\left(\begin{array}{c}n \\ M\end{array}\right)}
    \end{eqnarray}
    with $\mu' =2M'-n$, where $ M'\in \mathbb{Z}^+$ and $(n+2)r-1\leq M' < (n+2)r$. Now, there are two cases to consider;
    \begin{itemize}
        \item Case 1: $M'=M$. In this case, $r=1$ is not possible, whereas for $r=0$, the only possible way is when $M'=M=0$. Then 
        \begin{eqnarray}
            \frac{P(n+1)}{P(n)} = (1-r)(n+1)\frac{M!(n-M)!}{M'!(n+1-M')!} = \begin{cases}
                \frac{n+1}{n+1} =1 &\text{for $r=0$} \\
                \frac{(1-r)(n+1)}{n+1-M} < \frac{(1-r)(n+1)}{(n+1)(1-r)}=1 &\text{for $0<r<1$} 
            \end{cases}
        \end{eqnarray}
        \item Case 2: $M'=M+1$. In this case, $r=0$ is not possible, whereas for $r=1$, the only possible way is when $M=n$ and $M'=n+1$. Then 
        \begin{eqnarray}
            \frac{P(n+1)}{P(n)} = r(n+1)\frac{M!(n-M)!}{M'!(n+1-M')!} = \begin{cases}
                \frac{r(n+1)}{M+1} \leq \frac{r(n+1)}{((n+1)r-1)+1}=1 &\text{for $0<r<1$} \\
                \frac{(n+1)}{M+1} =\frac{(n+1)}{n+1} =1 &\text{for $r=1$}
            \end{cases}
        \end{eqnarray}
    \end{itemize}
    \noindent Since for both cases, the ratio of successive terms of $n$ is either less than, less than equals to or equals to 1, we have shown that at the expectation value of the binomial distribution where its probability is maximal, the probability decreases for increasing values of $n$.
\end{proof}

Thus, Eve's best strategy will be to input $n_R^-$ number of photons such that the guessing probability is maximized over the range $[n_R^-,n_R^+]$. The range of $\mathcal{X}^{\text{SDI}}_{r_0}$ considering the width of the ENOB voltage bin, $\left\lceil \delta V/\alpha_D \right\rceil$, spread equally around the peak of $x$ is given by
    \begin{equation}
        \mathcal{X}^{\text{SDI}}_{r_0} \in \mathbb{Z} \cap \left[\mu_x- \left\lceil\frac{\delta V}{2 \alpha_D} \right\rceil , \mu_x+ \left\lfloor \frac{\delta V}{2 \alpha_D} \right\rfloor \right].
    \end{equation}

\noindent In principle, $r_0$ can be chosen to be arbitrarily small. To ensure that we can approximate from a binomial distribution to a normal distribution, we will consider $n_R^- >10^5$ to be sufficiently large, as well as $r_0n_R^- >5$ and $(1-r_0)n_R^- >5$. Then $p_{\text{guess}}$ can be approximated by making a change of variable, where we let $n_A^- = (n_R^-+x)/2$, with a mean of $\mu_A^- = r_0n_R^-$ and a variance of $\sigma_A^2 = r_0(1-r_0)n_R^-$. The summation about the ENOB voltage bin width becomes an integral and $p_{\text{guess}}$ becomes
\begin{eqnarray}
    p_{\text{guess}} &&
    \leq  \frac{1}{2} \left[\text{erf}\left(\frac{\frac{\delta V}{2 \alpha_D}}{\sqrt{2\sigma^2_A}}\right) - \text{erf}\left(\frac{\frac{-\delta V}{2 \alpha_D}-1}{\sqrt{2\sigma^2_A}}\right)\right]
\end{eqnarray}

\noindent Therefore, for $m$ rounds of measurement,
\begin{eqnarray} \label{eq: Hmin equation}
    H_{\min,r_0}^{\text{SDI}} (X|E) \geq \kappa && = - m\log_2\left( p_{\text{guess}}\right) \nonumber \\
    && \geq - m\log_2\left(  \frac{1}{2} \left[\text{erf}\left(\frac{\frac{\delta V}{2 \alpha_D}}{\sqrt{2\sigma^2_A}}\right) - \text{erf}\left(\frac{\frac{-\delta V}{2 \alpha_D}-1}{\sqrt{2\sigma^2_A}}\right)\right]\right)
\end{eqnarray}

\noindent This completes the proof for $H_{\min,r_0}^{\text{SDI}} (X|E)$~\footnote{Note that the form presented here is different from Eq.C10 in Ref.~\cite{drahi2020certified} when $r_0=0.5$, as their conditional min-entropy has a typo with a missing factor of $\sqrt{2}$ in the denominator.}. By assuming the worst case, the lower bound of $\kappa$ will always be used to estimate the conditional min-entropy for the certified randomness generated. 

\section{Explicit form of \textit{Completeness}}\label{appendix: Completeness-Derivation}

Calculating the probability of certification test, $1-\epsilon_C$, requires the use of Eq.~\eqref{eq: completeness of coherent state}. However, this equation lacks an explicit form suitable for numerical computation. Therefore, this section attempts to present an explicit formulation for Eq.~\eqref{eq: completeness of coherent state}. From Appendix~\ref{appendix: Practical-Implementation}, the \textit{Completeness} of the SDI protocol with a coherent source input is defined as follows:
\begin{eqnarray}
    1-\epsilon_C &&= \text{tr}\left\{ \sum_{i_C=i_C^-}^{i_C^+}\ket{\alpha}\bra{\alpha}\hat{V}_{C}^{\sigma_{\gamma_C},\Delta_{\text{ADC}}}(i_C)\right\} \nonumber \\
    && = \text{tr}\left\{ \sum_{i_C=i_C^-}^{i_C^+}\ket{\alpha}\bra{\alpha}\int_{I_i^C}\hat{V}_{C}^{\sigma_C}(v_C)dv_C\right\} \nonumber \\
    && = \text{tr}\left\{ \sum_{i_C=i_C^-}^{i_C^+}\int_{v_C=L_a}^{L_b}\frac{e^{-\gamma_C^2/(2\sigma_{\gamma_C}^2)}}{\sqrt{2\pi}\sigma_{\gamma_C}}\sum_{n_c=n_C^{\min}}^{n_C^{\max}}\ket{\alpha}\bra{\alpha}\hat{N}_C(n_C)dv_C\right\}
\end{eqnarray}

\noindent where
\begin{eqnarray}
    I_i^C = \left[ \underbrace{\delta V_C\left(i_C-\frac{1}{2}\right)}_{L_a}, \underbrace{\delta V_C\left(i_C+\frac{1}{2}\right)}_{L_b} \right]. 
\end{eqnarray}

\noindent Since
\begin{eqnarray}
    \text{tr}\left\{\sum_{n_C=n_C^{\min}}^{n_C^{\max}}\ket{\alpha}\bra{\alpha}\hat{N}_C(n_C)\right\} = \sum_{n_C=n_C^{\min}}^{n_C^{\max}}\frac{e^{-\bar{n}_C}(\bar{n}_C)^{n_C}}{n_C!} \xrightarrow{\text{Gaussian}} \int_{n_C=n_C^{\min}}^{n_C^{\max}}\frac{e^{-(n_C-\bar{n}_C)^2/(2\bar{n}_C)}}{\sqrt{2\pi\bar{n}_C}}dn_C
\end{eqnarray}

\noindent where $\overline{n}_C$ is the mean photon number of $n_C$ and we approximate the probability distribution of the coherent source from Poisson to Gaussian since we consider $n_C>10^5$ to be sufficiently large. For consistency with the units, we will convert $n_C$ to $\alpha_C n_C$ to express everything here in terms of voltage. This gives $\alpha_C n_C \sim \mathcal{N}(\alpha_C \bar{n}_C,\bar{n}_C\alpha_C^2)$, where we will denote $\mu_{n_C} = \alpha_C \bar{n}_C$ and $\sigma_{n_C}^2= \alpha_C^2\bar{n}_C$. Then, we have
\begin{eqnarray}
    1-\epsilon_C = \sum_{i_C=i_C^-}^{i_C^+}\int_{v_C=L_a}^{L_b}\int_{n_C=n_C^{\min}}^{n_C^{\max}}\frac{e^{-\gamma_C^2/(2\sigma_{\gamma_C}^2)}}{\sqrt{2\pi}\sigma_{\gamma_C}}\frac{e^{-(\alpha_C n_C-\mu_{n_C})^2/(2\sigma_{n_C}^2)}}{\sqrt{2\pi}\sigma_{n_C}}dn_Cdv_C
\end{eqnarray}\\

The two exponents within the integral of $n_C$ can be further reduced to form a sum of two independent normal distributions for $v_C = \gamma_C + \alpha_C n_C $, where the probability distribution of $v_C \sim \mathcal{N}(0+\mu_{n_C},\sigma_{\gamma_C}^2+\sigma_{n_C}^2)$. To show this, we will first assume that $n_{C}^{\min}\ll\mu_{n_C}\ll n_{C}^{\max}$, as this is set to achieve optimal performance for the QRNG, as well as to prevent saturation at the certification photodetector. This allows us to do the following approximation
\begin{eqnarray}
\int_{n_C=n_C^{\min}}^{n_C^{\max}}\frac{e^{-(\alpha_C n_C-\mu_{n_C})^2/(2\sigma_{n_C}^2)}}{\sqrt{2\pi}\sigma_{n_C}}dn_C 
\approx&& 
\int_{n_C=-\infty}^{\infty}\frac{e^{-(\alpha_C n_C-\mu_{n_C})^2/(2\sigma_{n_C}^2)}}{\sqrt{2\pi}\sigma_{n_C}}dn_C =1.
\end{eqnarray}
Subsequently, using this approximation, the probability distribution for $v_C$ can be obtained as follows.
\begin{eqnarray}
\int_{n_C=n_C^{\min}}^{n_C^{\max}}\frac{e^{-\gamma_C^2/(2\sigma_{\gamma_C}^2)}}{\sqrt{2\pi}\sigma_{\gamma_C}}\frac{e^{-(\alpha_C n_C-\mu_{n_C})^2/(2\sigma_{n_C}^2)}}{\sqrt{2\pi}\sigma_{n_C}}dn_C 
\approx&& 
\int_{n_C=-\infty}^{\infty}\frac{e^{-\gamma_C^2/(2\sigma_{\gamma_C}^2)}}{\sqrt{2\pi}\sigma_{\gamma_C}}\frac{e^{-(\alpha_C n_C-\mu_{n_C})^2/(2\sigma_{n_C}^2)}}{\sqrt{2\pi}\sigma_{n_C}}dn_C \nonumber \\
=&&\frac{e^{-(v_C-\mu_{v_C})^2/(2\sigma_{v_C}^2)}}{\sqrt{2\pi}\sigma_{v_C}}
\end{eqnarray}

\noindent where we use the convolution proof of the sum of two normal independent random variables and we have $\mu_{v_C} = \mu_{n_C} =\alpha_C \bar{n}_C$ and $\sigma_{v_C}^2 = \sigma_{\gamma_C}^2+\sigma_{n_C}^2$ for completing the squares in the intermediate steps. Thus, the explicit form of \textit{Completeness} is given by
\begin{eqnarray} \label{eq: explicit cert test P}
    1-\epsilon_C && = \sum_{i_C=i_C^-}^{i_C^+}\int_{v_C=L_a}^{L_b} \frac{e^{-(v_C-\mu_{v_C})^2/(2\sigma_{v_C}^2)}}{\sqrt{2\pi}\sigma_{v_C}} dv_C \nonumber \\
    && = \sum_{i_C=i_C^-}^{i_C^+} \frac{1}{2}\left[\text{erf}\left(\frac{\delta V_C(i_C+\frac{1}{2})-\mu_{v_C}}{\sqrt{2}\sigma_{v_C}}\right)- \text{erf}\left(\frac{\delta V_C(i_C-\frac{1}{2})-\mu_{v_C}}{\sqrt{2}\sigma_{v_C}}\right)\right]
\end{eqnarray}

\noindent where $\delta V_C = (V_{C,\max}-V_{C,\min})/2^{\Delta_{\text{ADC}}}$. Therefore, if the voltage measurement at the certification photodetector exhibits a Gaussian distribution, then the \textit{Completeness} of the SDI protocol can be evaluated using Eq.~\eqref{eq: explicit cert test P}. 

\section{Mathematical Details for Unbalanced Device-Dependent QRNG Protocol} \label{appendix: Math-Details-for-Unbalanced-Homodyne-Detection}
 
\subsection{Unbalanced Homodyne Detection}

When the beam splitter is unbalanced, it causes an imperfect cancellation of the local oscillator fluctuation during the detection process. Due to this, the voltage variance of the local oscillator fluctuation, $\sigma_{\text{LO},V}^2$, is then mixed and captured together with the fluctuation of the vacuum signal, $\sigma_{Q,V}^2$, and the electronic noise of the photodetectors, $\sigma_{\gamma}^2$. The resultant total voltage variance measured at the unbalanced homodyne detection, $\sigma_{\text{UHD},V}^2$, is given by
\begin{eqnarray}
    \sigma_{\text{UHD},V}^2 = \sigma_{\text{LO},V}^2 + \sigma_{Q,V}^2  + \sigma_{\gamma}^2 
\end{eqnarray}
\noindent where these variance terms become independent of each other in the linear regime \cite{huang2020practical}. Both $\sigma_{\text{LO},V}^2$ and $\sigma_{\gamma}^2$ can be measured experimentally, but the experimental value of $\sigma_{Q,V}^2$ can only be obtained by subtracting them from $\sigma_{\text{UHD},V}^2$. To understand how each variance is theoretically obtained in the unbalanced homodyne model, we derive $\sigma_{Q,V}^2$ and $\sigma_{\text{LO},V}^2$ using the quadrature formalism in the shot-noise unit \cite{laudenbach2018continuous}. 

As illustrated in Fig.~\ref{fig: unbalanced homodyne}, the output $\hat{a}_A$ and $\hat{a}_B$ from port A and B, respectively, after an arbitrary beam splitter with reflectivity $r_0$, using a local oscillator $\hat{a}_{\text{LO}}$ and a signal $\hat{a}_{S}$ as input, is given by
\begin{eqnarray}
    \begin{pmatrix}
        \hat{a}_A \\
        \hat{a}_B
    \end{pmatrix} = 
    \begin{pmatrix}
        \sqrt{r_0} & \sqrt{1-r_0}\\
        -\sqrt{1-r_0} & \sqrt{r_0}
    \end{pmatrix}
    \begin{pmatrix}
        \hat{a}_{\text{LO}}\\
        \hat{a}_S
    \end{pmatrix}
    \quad \Rightarrow  \quad
    \begin{aligned}
        \hat{a}_A &= \sqrt{r_0}\hat{a}_{\text{LO}} + \sqrt{1-r_0}\hat{a}_S \\
        \hat{a}_B &= -\sqrt{1-r_0}\hat{a}_{\text{LO}}+\sqrt{r_0}\hat{a}_S.
    \end{aligned}
\end{eqnarray}
The difference of the photon number in the outputs, denoted by $\hat{N}_-$, is 
\begin{eqnarray}
    \hat{N}_- &&= \hat{a}_A^\dagger \hat{a}_A - \hat{a}_B^\dagger \hat{a}_B \nonumber \\
    &&=(2r_0-1)(\hat{n}_{\text{LO}} -\hat{n}_{S})+ 2\sqrt{1-r_0}\sqrt{r_0}(\hat{a}_S^\dagger\hat{a}_{\text{LO}}+\hat{a}_{\text{LO}}^\dagger\hat{a}_S)
\end{eqnarray}

\begin{figure}[t!]
    \centering
    \includegraphics[width=0.3\textwidth]{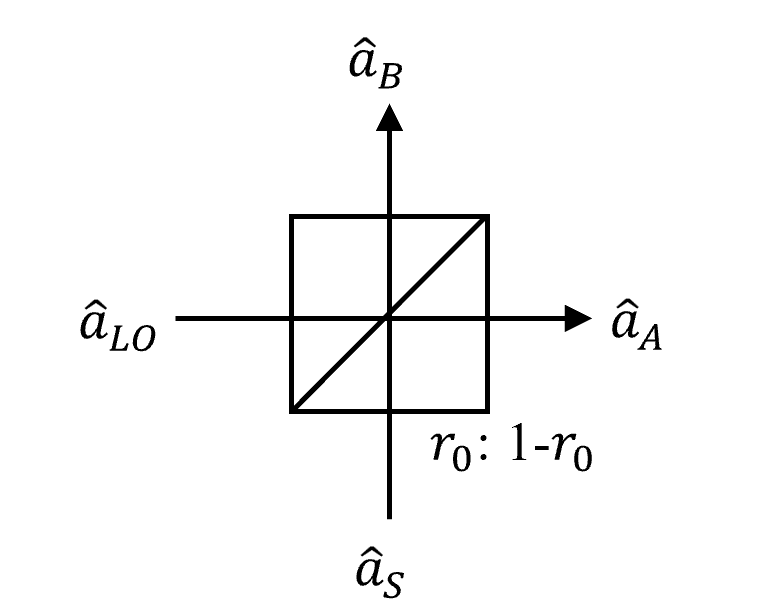}
    \caption{Unbalanced homodyne detection in quadrature formalism}
    \label{fig: unbalanced homodyne}
\end{figure}

\noindent where $\hat{n}_{\text{LO}} = \hat{a}_{\text{LO}}^\dagger\hat{a}_{\text{LO}}$ and $\hat{n}_{S} = \hat{a}_{S}^\dagger\hat{a}_{S}$. Since the local oscillator is a single mode and with an \textit{a priori} understanding that the signal amplitude is small with respect to the local oscillator amplitude, that is, $\alpha_S \ll \alpha_{\text{LO}}$, then using the formalism in Ref.~\cite{bachor2019guide}, $\hat{a}_{\text{LO}}$ can be decomposed into $\hat{a}_{\text{LO}} = \alpha_{\text{LO}}+ \delta\hat{a}_{\text{LO}}$, where $\alpha_{\text{LO}}$ is the mean amplitude of the local oscillator and $\delta\hat{a}_{\text{LO}}$ represents a small quadrature amplitude fluctuation about $\alpha_{\text{LO}}$. The number operator of the local oscillator can be rewritten as
\begin{eqnarray}
    \hat{n}_{\text{LO}} &&\approx \overline{n}_{\text{LO}} + \alpha_{\text{LO}} (\delta\hat{a}_{\text{LO}}^\dagger + \delta\hat{a}_{\text{LO}}) \nonumber \\
    && = \overline{n}_{\text{LO}} +\alpha_{\text{LO}}\delta\hat{x}_{\text{LO}}
\end{eqnarray}
\noindent where $\overline{n}_{\text{LO}}$ is the mean photon number of the local oscillator, $\delta\hat{x}_{\text{LO}} = \delta\hat{a}_{\text{LO}}^\dagger + \delta\hat{a}_{\text{LO}}$ is the amplitude fluctuation quadrature operator of the local oscillator and we ignore the second-order $\delta$ terms. With this, $\hat{N}_-$ can be simplified to
\begin{eqnarray}
    \hat{N}_- \approx (2r_0-1)(\overline{n}_{\text{LO}}+\alpha_{\text{LO}}\delta\hat{x}_{\text{LO}} - \hat{n}_{S})+ 2\sqrt{1-r_0}\sqrt{r_0}(\hat{a}_S^\dagger\hat{a}_{\text{LO}}+\hat{a}_{\text{LO}}^\dagger\hat{a}_S).
\end{eqnarray}
The variance of $\hat{N}_-$ is then given by
\begin{eqnarray}
    \text{Var}(\hat{N}_-) &&= \langle\hat{N}_-^2\rangle_{\alpha_{\text{LO}}} -\langle\hat{N}_-\rangle^2_{\alpha_{\text{LO}}} \nonumber \\
    && = \underbrace{(2r_0-1)^2\text{Var}(\hat{n}_{\text{LO}})}_{\sigma_{\text{LO}}^2 } + \underbrace{4r_0(1-r_0)(\overline{n}_{\text{LO}}\text{Var}(\hat{x}_S)+\overline{n}_S)}_{\sigma_Q^2}   
\end{eqnarray}

\noindent where $\langle\hat{N}_-\rangle_{\alpha_{\text{LO}}} = \text{Tr}\{\hat{N}_-(\hat{\rho}_S \otimes \ket{\alpha_{\text{LO}}}\bra{\alpha_{\text{LO}}})\}$ with a coherent local oscillator, $\langle\hat{a}_S^\dagger\hat{a}_{\text{LO}} + \hat{a}_{\text{LO}}^\dagger\hat{a}_S\rangle \approx \alpha_{\text{LO}}\langle\hat{x}_S\rangle$, $\langle \hat{a}_S^\dagger+\hat{a}_S \rangle= \langle\hat{x}_S\rangle$, where $\hat{x}_{S}$ is the amplitude quadrature of the signal, $\text{Var}(\delta\hat{x}_{\text{LO}}) = \langle(\delta\hat{x}_{\text{LO}})^2\rangle - \langle\delta\hat{x}_{\text{LO}}\rangle^2, \text{Var}(\hat{x}_{S}) = \langle\hat{x}_{S}^2\rangle - \langle\hat{x}_{S}\rangle^2$ and $\text{Var}(\hat{n}_{\text{LO}})=\overline{n}_{\text{LO}}\text{Var}(\delta\hat{x}_{\text{LO}})$. The variance from the difference measurement is made up of two contributions: (i) the variance of the fluctuation of the local oscillator, $\sigma_{\text{LO}}^2 = (2r_0-1)^2\text{Var}(\hat{n}_{\text{LO}})$ and (ii) the variance of the fluctuation of the signal, $\sigma_{Q}^2 = 4r_0(1-r_0)(\overline{n}_{\text{LO}}\text{Var}(\hat{x}_S)+\overline{n}_S)$. We can represent the ratio of the fluctuation of the local oscillator to its mean photon number as $f=\sqrt{\text{Var}(\hat{n}_{\text{LO}})}/\Bar{n}_{\text{LO}}$ \cite{chi2011balanced,PhysRevA.98.012312}.

For our case, since our signal is vacuum, we have $\text{Var}(\hat{x}_{S})=1$ in the shot-noise unit and $\overline{n}_S=0$ for $\sigma_Q^2$. Hence, using $\sigma^2_{\text{UHD}}$ to represent the unbalanced homodyne detection in our experimental setup, we have 
\begin{eqnarray}
    \sigma^2_{\text{UHD}} = (2r_0-1)^2f^2 \overline{n}^2_{\text{LO}} + 4r_0(1-r_0)\overline{n}_{\text{LO}} +\sigma^2_{n_\gamma}
\end{eqnarray}
where we have included the variance of the photodetector's electronic noise, $\sigma^2_{n_\gamma}$, in terms of photon number. Interestingly, $\sigma^2_{\text{UHD}}$ depends on $\overline{n}_{\text{LO}}$ quadratically and linearly due to the fluctuation of the local oscillator and the vacuum, respectively. Lastly, its voltage variance is given by
\begin{eqnarray}
    \sigma_{\text{UHD},V}^2 = \alpha_D^2 \left((2r_0-1)^2f^2 \overline{n}^2_{\text{LO}}\right) + \alpha_D^2\left(4r_0(1-r_0)\overline{n}_{\text{LO}}\right) + \sigma_{\gamma}^2.
\end{eqnarray}

\subsection{Conditional Min-Entropy of Unbalanced Device-Dependent Protocol}

The conditional min-entropy of the DD homodyne protocol is $H_{\min,r_0}^{\text{DD}}(X|E) = -\log_2\left[\max\left(c_1,c_2\right)\right]$ \cite{haw2015maximization,huang2020practical,zheng20196}, where

\begin{eqnarray}
    c_1 = \frac{1}{2}\left[\text{erf}\left(\frac{\gamma_{D,\max}-V_{\max} +3\delta V/2 }{\sqrt{2\sigma_{Q,V}^2}}\right)+1\right] \quad \text{and} \quad c_2 = \text{erf}\left(\frac{\delta V/2}{\sqrt{2\sigma_{Q,V}^2}}\right)
\end{eqnarray}

\noindent with $\sigma_{Q,V}^2 = \alpha_D^2 4r_0(1-r_0)\overline{n}_R$ and $\overline{n}_R$ is the average photon number of $n_R$. $c_1$ is the probability when the voltage measurement outcome of $x$ is the highest point at the saturation limits of the ADC sampling range, while $c_2$ is the probability of the voltage measurement outcome of $x$ is the highest at its mean. In our experimental setup, the sampling range of the ADC is fixed from $V_{\min}=-1$V to $V_{\max}=1$V, and the voltage measurement of $x$ is much lower than $V_{\max}$ and much higher than $V_{\min}$, indicating that $x$ will always be within the sampling range. Moreover, the variance of our electronic noise from the AC coupled balanced detector $\gamma^2_D$ is measured to be very small with respect to $\sigma^2_{Q,V}$. Thus, we can safely assume that $c_1 \leq c_2$ and simplifies our analysis for $H_{\min,r_0}^{\text{DD}}(X|E)$ using $c_2$. To remain consistent with $H_{\min,r_0}^{\text{SDI}}(X|E)$, we modify $H_{\min,r_0}^{\text{DD}}(X|E)$ in terms of photon number, and it can be rewritten as 

\begin{eqnarray}\label{eq: HminDD c2}
     H_{\min,r_0}^{\text{DD}}(X|E) = -\log_2\left[\frac{1}{2}\left(\text{erf}\left(\frac{\frac{\delta V}{2\alpha_D}}{\sqrt{2\sigma_{Q}^2}}\right) - \text{erf}\left(\frac{\frac{-\delta V}{2\alpha_D}-1}{\sqrt{2\sigma_{Q}^2}}\right)\right)\right]
\end{eqnarray}  
where $\sigma_{Q}^2 = 4r_0(1-r_0)\overline{n}_R$.

\section{Comparison between unbalanced DD and extended SDI protocol} \label{appendix: HminDD-HminSDI}

The purpose of this analysis is to compare the difference between the randomness generated in the unbalanced DD and the extended SDI protocol during practical implementation. The difference in the ideal photon-counting ADC case will be considered first, followed by their experimental difference.

\subsection{Ideal photon-counting ADC}

Assuming an ideal ADC that could distinguish between $n$ and $n+1$ photons, the width of this ADC is set to be $\delta V_0/\alpha_D =1$ photon wide. Then, the general expression for the min-entropy of both protocols is given by

\begin{eqnarray}\label{eq: Hmin ideal ADC}
     H_{\min,r_0}^{\text{protocol}}(X|E) = -\log_2\left[\frac{1}{2}\left(\text{erf}\left(\frac{\frac{\delta V_0}{2\alpha_D}}{\sqrt{2\sigma_{k}^2}}\right) - \text{erf}\left(\frac{\frac{-\delta V_0}{2\alpha_D}-1}{\sqrt{2\sigma_{k}^2}}\right)\right)\right].
\end{eqnarray}    

\noindent for $k\in\{Q,A\}$. In this case, we assume $\sigma_Q^2$ and $\sigma_A^2$ to be greater than $10^5$ as, in principle, $r_0$ can be arbitrarily small, and to simplify our analysis, we will also consider only $0.5\leq r_0 \leq 0.9$. This ensures that $\delta V_0/2\alpha_D \ll \sqrt{2\sigma_k^2}$ and $ -\delta V_0/2\alpha_D-1 \ll \sqrt{2\sigma_k^2}$, allowing us to perform a Taylor expansion to approximate the error function up to the first-order term. The unbalanced DD and extended SDI protocol for an ideal photon-counting ADC is given by
\begin{eqnarray}
    H^{\text{DD}}_{\min,r_0}(X|E) &&= -\log_2\left(\frac{2}{\sqrt{2\pi\sigma_Q^2}}\right) = \frac{1}{2}\log_2\left(2\pi (4r_0)(1-r_0)\overline{n}_R\right)-1 \\
    H^{\text{SDI}}_{\min,r_0}(X|E) &&\geq -\log_2\left(\frac{2}{\sqrt{2\pi\sigma_A^2}}\right) 
    = \frac{1}{2}\log_2\left(2\pi r_0(1-r_0)n_R^-\right) -1.
\end{eqnarray}
Therefore, their difference is 
\begin{eqnarray} \label{eq: theoretical diff of HminDD and HminSDI} 
    \Lambda_{\text{ideal}} &&= H_{\min,r_0}^{\text{DD}}\left(X|E\right) - H_{\min,r_0}^{\text{SDI}}\left(X|E\right) \geq 1+\frac{1}{2}\log_2\left(\frac{\overline{n}_R}{n_R^-}\right).
\end{eqnarray}
Since $\overline{n}_R>n_R^-$, the above relation will always yield a value greater than $1$ bit. 

\subsection{Comparison between the experimental difference}

Based on the experimental parameters in Sec.~\ref{sec: Extended SDI Protocol Analysis}, we have $ \delta V/2\alpha_D\gg1$, and we can do the following approximation: $ -\delta V/2\alpha_D-1\approx -\delta V/2\alpha_D$. This results in the following experimental min-entropy. 
\begin{eqnarray}\label{eq: Hmin experimental ADC}
     H_{\min,r_0}^{\text{protocol}}(X|E) \approx -\log_2\left[\text{erf}\left(\frac{\frac{\delta V}{2\alpha_D}}{\sqrt{2\sigma_{k}^2}}\right)\right].
\end{eqnarray}    
The experimental difference is given by
\begin{eqnarray}\label{eq: HminDD-HminSDI difference}
    \Lambda_{\text{ENOB},r_0} &&= H_{\min,r_0}^{\text{DD}}\left(X|E\right) - H_{\min,r_0}^{\text{SDI}}\left(X|E\right) \nonumber \\[0.5cm]
    &&\gtrsim -\log_2\left[\text{erf}\left(\frac{\frac{\delta V}{2\alpha_D}}{\sqrt{2\sigma_Q^2}} \right)\right] - \left(-\log_2 \left[\text{erf}\left(\frac{\frac{\delta V}{2\alpha_D}}{\sqrt{2\sigma_A^2}} \right)\right]\right) \nonumber \\[0.5cm]
    && = 1+\frac{1}{2}\log_2\left(\frac{\overline{n}_R}{n_R^-}\right) + \log_2\left[\frac{\sqrt{n_R^-}\text{erf}\left(\frac{\frac{\delta V}{2\alpha_D}}{\sqrt{2\sigma_A^2}} \right)}{\sqrt{4\overline{n}_R}\text{erf}\left(\frac{\frac{\delta V}{2\alpha_D}}{\sqrt{2\sigma_Q^2}} \right)}\right]
\end{eqnarray}
The first two terms of $\Lambda_{\text{ENOB},r_0}$ are the result of $\Lambda_{\text{ideal}}$ and the last term is due to the contribution of the ENOB. In the asymptotic limit, when both $\Bar{n}_R$ and $n_R^-$ tend towards infinity, we have
\begin{eqnarray} \label{eq: Lambda ENOB to 1 bit}
    \lim_{\overline{n}_R,n_R^-\rightarrow\infty} \Lambda_{\text{ENOB},r_0}=&&\lim_{\overline{n}_R,n_R^-\rightarrow\infty} \left(\Lambda_{\text{ideal}} + \log_2\left[\frac{\sqrt{n_R^-}\text{erf}\left(\frac{\frac{\delta V}{2\alpha_D}}{\sqrt{2\sigma_A^2}} \right)}{\sqrt{4\overline{n}_R}\text{erf}\left(\frac{\frac{\delta V}{2\alpha_D}}{\sqrt{2\sigma_Q^2}} \right)}\right]\right) \nonumber \\
    = &&1 + \lim_{\overline{n}_R,n_R^-\rightarrow\infty} \log_2\left(\frac{\overline{n}_R}{n_R^-}\right)\nonumber \\
    =&&1 
\end{eqnarray}
where in the first line, we approximate the error functions in the last term of $\Lambda_{\text{ENOB},r_0}$ using the Taylor expansion up to its first-order term, as both $(\delta V/2\alpha_D)/\sqrt{2\sigma_A^2} \ll 1$ and $(\delta V/2\alpha_D)/\sqrt{2\sigma_Q^2} \ll 1$. After expansion, the term inside the logarithm becomes $1$ and $\log_2(1)=0$, where the ENOB contribution vanishes. This also indicates that the effect of ENOB vanished as both $\Bar{n}_R$ and $n_R^-$ increase, and having a sufficiently lower/higher ENOB bit depth will also result in the same outcome. Therefore, the two protocols have $1$ bit of difference in randomness at the asymptotic limit of large $\Bar{n}_R$ and $n_R^-$.

\section{FPGA Resources and Architecture} \label{appendix: FPGA-Architecture-and-Analysis}

\subsection{Choice of Hashing Block Size}

Understanding the usage of FPGA resources and the hashing security parameter $\epsilon_{\text{hash}}$ with different hashing block sizes of $l\times h$ is necessary to choose the optimal parameters for the real-time SDI-QRNG operation. For this analysis, we will use the same experimental parameter as Sec.~\ref{sec: Real-Time SDI-QRNG Performance}. The $H_{\min,r_0}^{\text{SDI}} =3.354$ bits per sample at $r_0=0.513$ give an upper bound for the compression ratio $r\leq H_{\min,r_0}^{\text{SDI}}/b = 23.96\%$. We have opted for a conservative compression ratio of $r\approx20\%$ for the real-time operation. 

\begin{figure}[ht]
    \centering
    
    \subfigure[FPGA resources utilization.]{\includegraphics[width=0.48\linewidth]{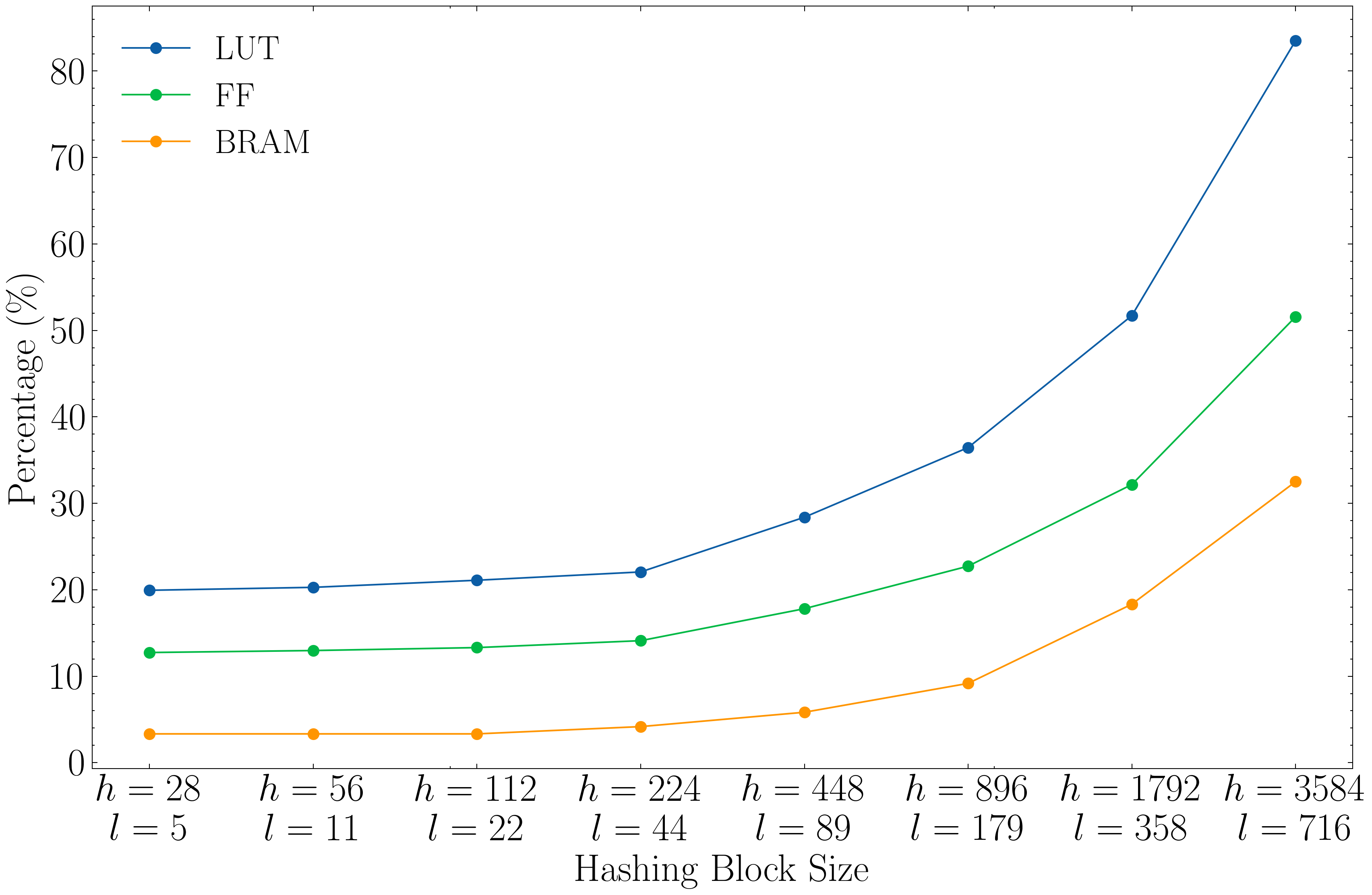}\label{fig: FPGA resources}}
    \centering
    \subfigure[Hashing failure probability $\epsilon_{\text{hash}}$.]{\includegraphics[width=0.48\linewidth]{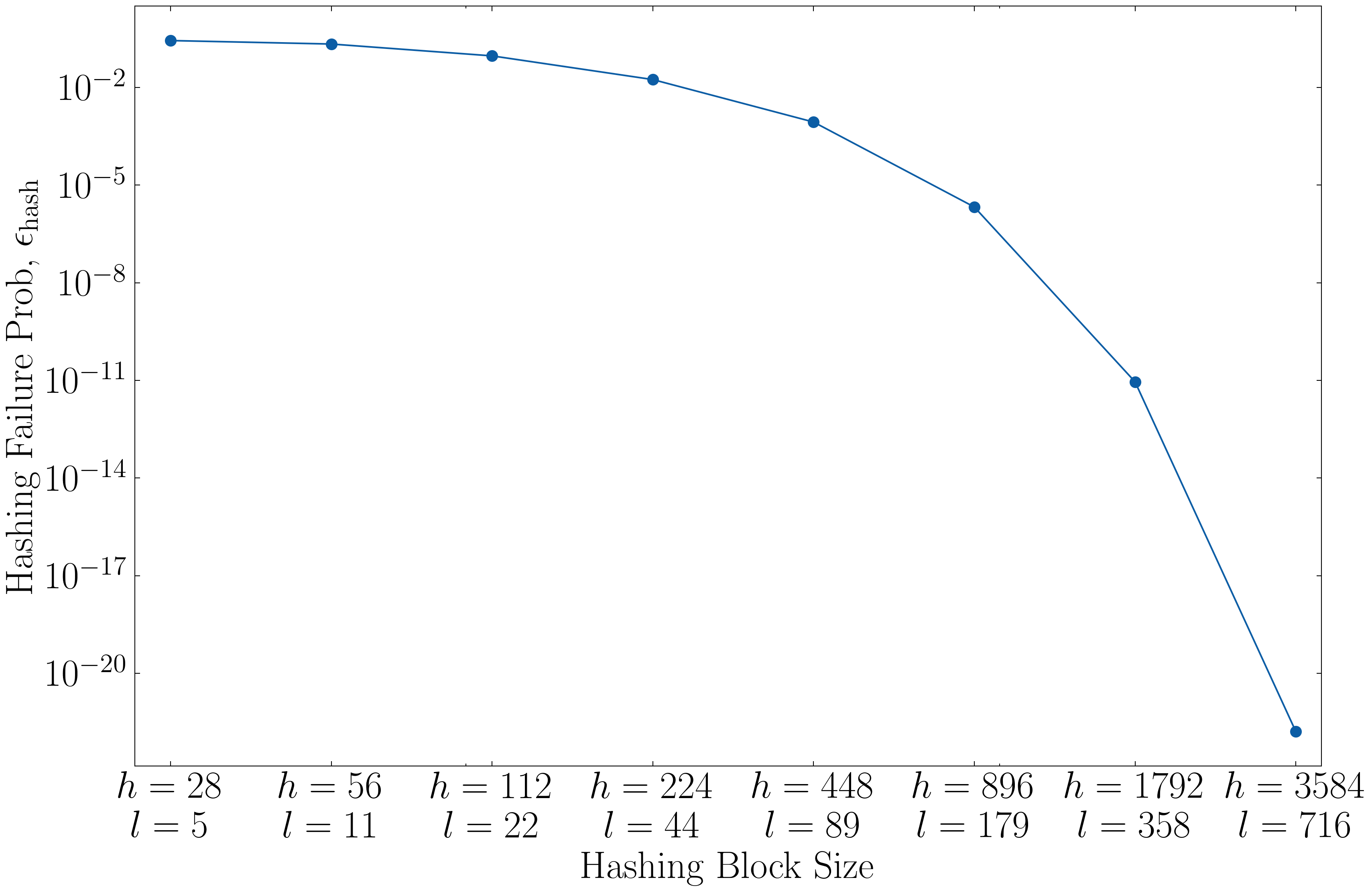}\label{fig: hashing block with ehash}}
    
    \caption{Overview of FPGA resources utilization and the hashing failure probability with different choices of hashing block sizes. The hashing block sizes are selected in the following manner: $h$ is chosen such that $h=2^k\times 14$, for $k\in [1,8]$, and $l$ is chosen such that $l/h \simeq 20\%$. (a): The FPGA resource utilization increases as the hashing block sizes increases, where the LUT is the first to be utilized. (b): The $\epsilon_{\text{hash}}$ decreases with increasing hashing block sizes, indicating that the larger block size is better for hashing.}
    \label{fig: FPGA resources compile}
\end{figure}

To find the optimal hashing parameters, we evaluate the usage of FPGA resources with different hashing block sizes~\cite{zheng20196}], as plotted in Fig.~\ref{fig: FPGA resources}. 
The FPGA resources to be analyzed are the programmable logics (Look-Up Tables (LUT) and Flip-Flops (FF)), as well as the Block Random Access Memory (BRAM) responsible for storing the Toeplitz matrix. The results of these parameters are obtained from the Vivado implementation report after designing the FPGA algorithm. This allows users to understand how well their algorithm will work on their FPGA board prior to deployment. In the Red Pitaya FPGA board, the total number of programmable logic available for LUT, FF, and BRAM are $17600$, $35200$, and $60$ ($2.1$Mb), respectively. The LUT is close to full utilization ($83.51\%$) when the hashing block size is $716 \times 3584$, whereas the FF and BRAM still have sufficient resources left. This illustrates that the LUT is the bottleneck in our FPGA board, and increasing the hashing block size any further will result in utilizing all the LUT first.

Moreover, the hashing security parameter, $\epsilon_{\text{hash}}$, must be small so that the overall composable security $\epsilon$ of the SDI-QRNG can also be small. For example, in Fig.~\ref{fig: hashing block with ehash}, to achieve a relatively small $\epsilon_{\text{hash}}$ such that $\epsilon$ can be lower than $10^{-10}$, the size of the hashing block must be at least $358 \times 1792$, while maintaining $r$ at approximately $20\%$. With these analysis done, the suitable security parameters and hashing block size for real-time SDI-QRNG can be determined.

Lastly, designing the FPGA algorithm with PYNQ has its own constraint when choosing the length of the output $l$. One of the hardware Intellectual Property (IP) that the PYNQ library supports to acquire the ADC measurement is the Direct Memory Access (DMA). In simple terms, the DMA manages the transfer of ADC measurements (in binary form) from the programmable logic (PL) fabric to the memory block in the processing system (PS). We note that our DMA is designed to transfer binary data only in a block size of $2^{k}$ bits, up to a maximum of $2^{10} = 1024$ bits. For example, if $l=358$ bits of binary are produced from the hashing, then a DMA block size of at least $2^9=512$ bits will be needed. Despite not filling up the block, the memory in the PS will still be allocated to receive $512$ bits of data rather than $358$ bits. With this in mind, maximizing the length of $l$ in each DMA block without wasting unnecessary PS memory resources is another factor to consider for our optimization. Taking the FPGA resource evaluation for different block sizes and hashing security parameter into account, we choose a block size of $512 \times 2562$, i.e. $l=512$ bits and $h=2562$ bits for our Toeplitz extractor. 

\subsection{FPGA Architecture}

 \begin{figure*}
    \centering
    \subfigure[]{\includegraphics[width=0.48\linewidth]{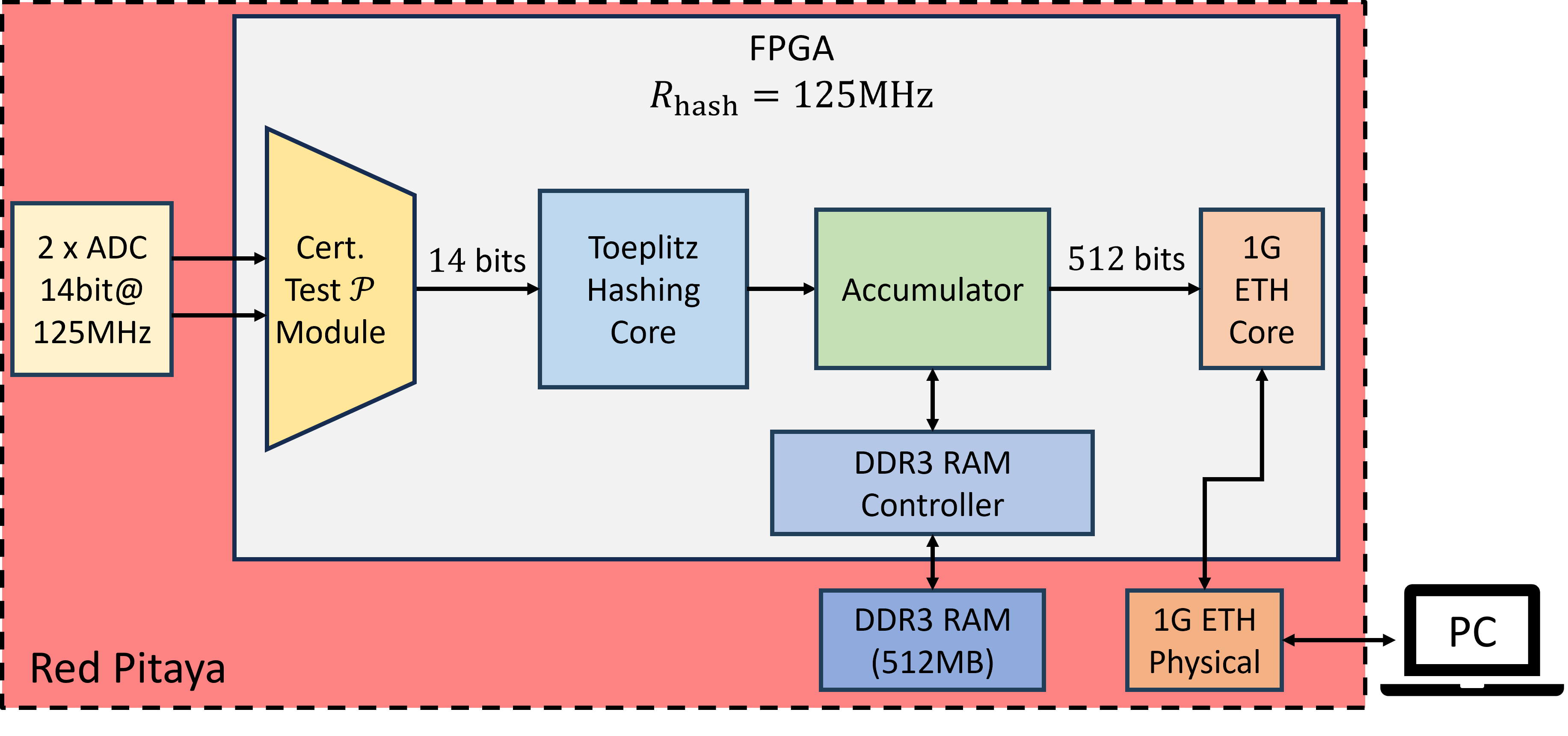}\label{fig: FPGA schematic}}
    \centering
    \subfigure[]{\includegraphics[width=0.49\linewidth]{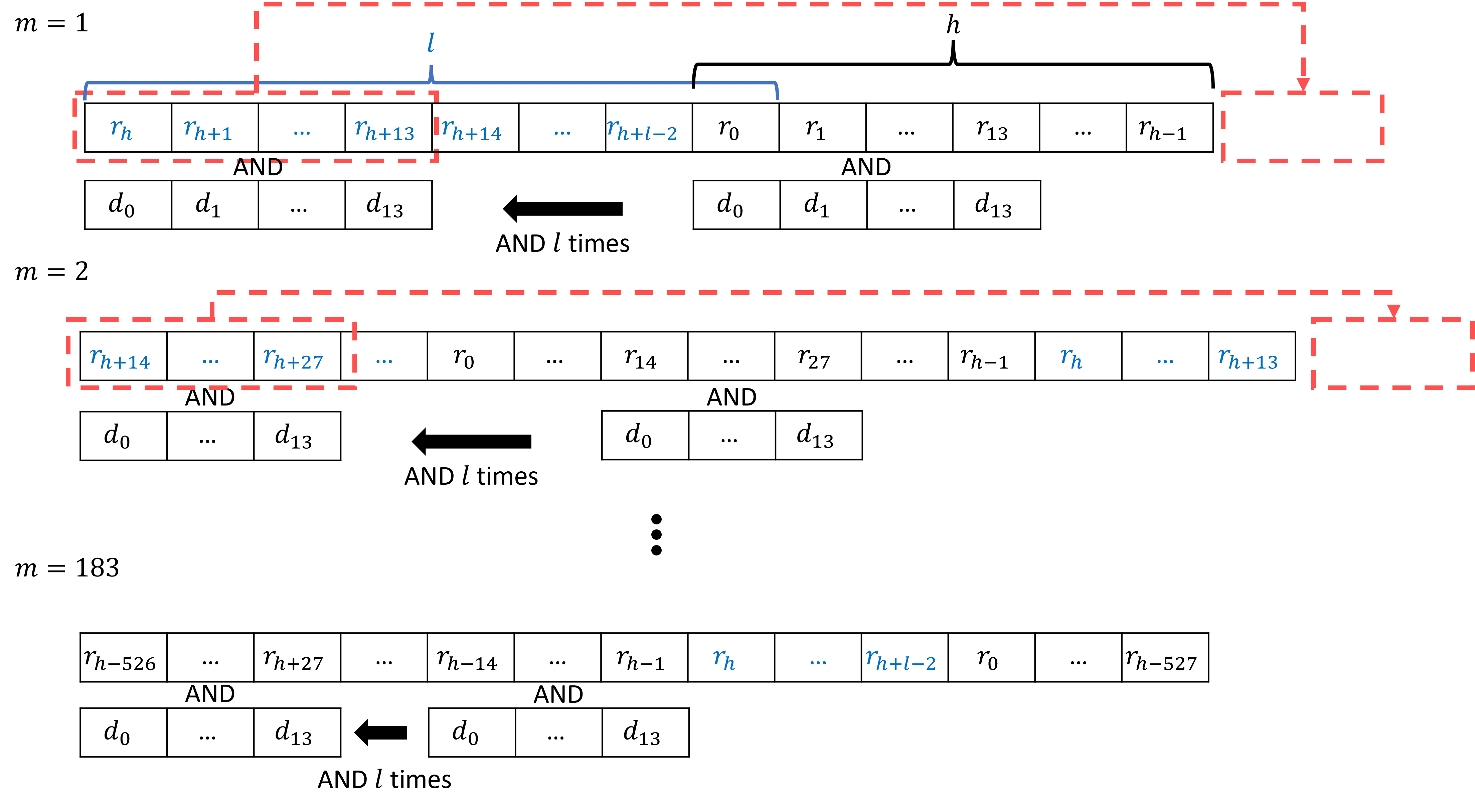}\label{fig: Toeplitz Hashing Schematic}}
    \caption{Overview of FPGA implementation (a): FPGA design schematic. DDR3: Double Date Rate 3, RAM: Random-Access Memory, ETH: Ethernet (b): Toeplitz hashing algorithm in the FPGA.}

\end{figure*}   

A high-level schematic of the FPGA architecture design is shown in Fig.~\ref{fig: FPGA schematic}. For $m=1$ round of measurement, samples from the two channels of the ADC, one for the certification measurement and the other for the randomness generation measurement, are sent into the certification test $\mathcal{P}$ module that assesses and rejects samples failing the test. Upon passing the test, it sends $b=14$ bits of the randomness generation measurement to the Toeplitz Hashing Core that performs randomness extraction. To save some on-chip memory resources in the FPGA, the matrix is represented in a binary string of length $l+h-1$. Subsequently, the $14$ bits entering will perform a bitwise "AND" operation with the corresponding $14$ bits substring from the Toeplitz binary string for $l=512$ times to produce a subhashed binary string at the end. This Toeplitz hashing process is illustrated in Fig.~\ref{fig: Toeplitz Hashing Schematic} for clarity. Afterwards, this subhashed string will enter the accumulator, and an "XOR" bitwise operation will be performed with other subhashed strings that were stored in the DDR3 RAM for the next few rounds. This process will repeat for $m=183$ rounds to produce the final hashed random bits of length $l$ and will be transferred to the computer via the Ethernet cable. This completes $t=1$ cycle of hashing. 

\section{Application for SDI-QRNG with ASE light source} \label{appendix: ASE-Based-SDI-QRNG-Theory}

The theoretical description of the ASE source is presented as follows. For one mode of ASE source\footnote{Note that the modes mentioned here do not have the same meaning as the "single mode" in single mode fiber. Instead, the modes here simply mean the number of degeneracy, $M$, in the ASE source.}, its photon statistics is equivalent to that of a thermal state, where it can be described by the Bose-Einstein distribution \cite{martin2015quantum,wong1998photon,pietralunga2003photon,yang2020randomness,li2021experimental}

\begin{equation}
    P\left(n,\Bar{n}\right) = \frac{\Bar{n}^n}{(1+\Bar{n})^{1+n}}
\end{equation}

\noindent where $P\left(n,\Bar{n}\right)$ is the probability of counting $n$ photons and $\Bar{n}$ is the average number of photons. Generally, an ASE source contains $M$ number of independent modes, is related to the ratio of its optical bandwidth $B_{\text{opt}}$ to the bandwidth of the photodetector $B_{\text{pd}}$ during detection. The photon statistics of the ASE source can be described by the $M$-fold degenerate Bose-Einstein distribution \cite{martin2015quantum,wong1998photon,pietralunga2003photon,yang2020randomness,li2021experimental}

\begin{equation}
    P\left(n,\Bar{n},M\right) = \frac{\Gamma(n+M)}{\Gamma(n+1)\Gamma(M)}\left(1+\frac{1}{\Bar{n}}\right)^{-n}\left(1+\Bar{n}\right)^{-M}
\end{equation}

\noindent where $\Gamma(\cdot)$ is the gamma function, $n$ is the number of photons per mode and $\Bar{n}$ is the average number of photons per mode. In addition, for an ASE source with a Gaussian power spectral density, $M$ is given by \cite{pietralunga2003photon,yang2020randomness,li2021experimental}

\begin{eqnarray}
    M  = s \frac{\pi\Tilde{B}^2}{\pi\Tilde{B}~\text{erf}(\sqrt{\pi} \Tilde{B})-\left[1-\text{exp}\left(-\pi \Tilde{B}^2\right)\right]}
\end{eqnarray}

\noindent where $\Tilde{B} = B_{\text{opt}}/B_{\text{pd}}$ and $s$ is the polarization degeneracy of the ASE source. For a polarized ASE source, we have $s=1$, while for an unpolarized ASE source, we have $s=2$. The average photon number for the ASE source with $M$ modes is denoted by $\Bar{n}_{\text{ASE},M}$. Thus, the number of photons in $M=1$ mode is

\begin{equation}
    \Bar{n}_{\text{ASE}} = \frac{\Bar{n}_{\text{ASE},M}}{M}
\end{equation}

\noindent with a variance of $\sigma_{\text{ASE}}^2 = \Bar{n}_{\text{ASE}} + (\Bar{n}_{\text{ASE}}^2/M)$ \cite{pietralunga2003photon}. Subsequently, by the sum of independent random variables, the variance of the ASE source with $M$ independent modes is \cite{li2021experimental}

\begin{equation}
    \sigma_{\text{ASE},M}^2 = \Bar{n}_{\text{ASE},M} + \Bar{n}_{\text{ASE}}^2.
\end{equation}

The characterization process of the ASE source can be performed by first placing a narrow bandpass filter before the measurement devices to ensure that the light entering is centered at $1550$nm and obtain the desired narrowband optical spectrum. This is done to ensure that the ASE source is operating in accordance with the assumption of the SDI protocol at $1550$nm. Subsequently, the voltage bound for the certification test $\mathcal{P}_{\text{ASE}}$ can be determined with an optical spectrum analyzer to measure $B_{\text{opt}}$ as it enters the certification photodetector. Once $B_{\text{opt}}$ is obtained, the number of $M$, $\Bar{n}_{\text{ASE}}$ and the variance $\sigma_{\text{ASE},M}^2$ could be obtained. With a large number of modes (usually for $M>100$) present in the ASE source, its photon distribution could be well modeled by a Gaussian distribution with a variance of $\sigma_{\text{ASE},M}^2$. This behavior has been verified separately in Ref.~\cite{yang2020randomness} and Ref.~\cite{li2021experimental}, and this usually holds in a higher optical power regime. Afterwards, the voltage bound of the certification test $\mathcal{P}_{\text{ASE}}$ could be obtained by using Eq.~\ref{eq: explicit cert test P}. 

\section{NIST Test Results} \label{appendix: NIST-Test-Results}

The NIST test results for the SDI-QRNG operation in Sec.~\ref{sec: Real-Time SDI-QRNG Performance} is presented in Table.~\ref{tab: NIST test}. $1$ Gbits of random binary data are collected and divided into 1000 sequences of 1 Mb for testing. The result shows that the random binary data successfully passed the NIST test suite.

\begin{table}[h]
\centering
\caption{NIST test results}
\begin{ruledtabular}
\begin{tabular*}{\linewidth}{@{\extracolsep{\fill}}lccc}
\multicolumn{4}{c}{\textbf{NIST Test}}\\ \toprule

Statistical Test & 
P-value & Proportion & Result
\\ \midrule

Frequency & 
0.073876 & 0.9910 & Pass \\

Block Frequency & 
0.257992 & 0.9920 & Pass \\      

Cumulative Sums & 
0.123324 & 0.9850 & Pass  \\

Runs & 
0.284119 & 0.9890 & Pass  \\ 

Longest Run & 
0.768789 & 0.9920 & Pass  \\

Rank & 
0.882397 & 0.9910 & Pass \\

FFT & 
0.126631 & 0.9900 & Pass  \\ 

Non-Overlapping Template & 
0.434772 & 0.9810 & Pass  \\ 

Overlapping Template  & 
0.145265 & 0.9910 & Pass \\

Universal & 
0.105580 & 0.9890 & Pass  \\

Approximate Entropy & 
0.124058 & 0.9860 & Pass  \\

Random Excursions & 
0.191052 & 0.9841 & Pass \\

Random Excursions Variant & 
0.826794 & 0.9825 & Pass  \\

Serial & 
0.898959 & 0.9850 & Pass  \\

Linear Complexity & 
0.776924 & 0.9900 & Pass  \\

\end{tabular*}
\end{ruledtabular}
\label{tab: NIST test}
\end{table}

\input{arxiv-bibitem}

\end{document}

%% file: arxiv-bibitem.tex
%